\documentclass[letter,12pt,english]{article}
\usepackage[utf8]{inputenc}   
\usepackage{lmodern}
\usepackage[english]{babel}   
\usepackage[babel]{microtype}
\usepackage{amsfonts,amsmath,latexsym,amsthm,fixmath}
\usepackage{amssymb,ifthen}
\usepackage{authblk}
\usepackage{enumitem}
\usepackage{color,graphicx}
\usepackage{url}
\usepackage[
    pdftex,
    pdfstartview={XYZ 0 1000 1.0},
    bookmarks=false,
    colorlinks,
    breaklinks=true,
    citecolor=citecolor,
    filecolor=blue,
    linkcolor=linkcolor,
    urlcolor=urlcolor,
    pdfborder={0 0 0}
    ]{hyperref}
\definecolor{urlcolor}{rgb}{0, 0.5, 0}
\definecolor{citecolor}{rgb}{.5,0,.25}
\definecolor{linkcolor}{rgb}{0,0,1}

\usepackage{geometry}
\newtheorem{theorem}{Theorem}
\newtheorem{lemma}[theorem]{Lemma}
\newtheorem{claim}{Claim}

\newcommand{\executeiffilenewer}[3]{%
\ifnum\pdfstrcmp{\pdffilemoddate{#1}}%
{\pdffilemoddate{#2}}>0%
{\immediate\write18{#3}}\fi%
}
\def\svgmode{ps}\ifx\pdfoutput\undefined
\def\executeiffilenewer#1#2#3{\immediate\write18{#3}}\else%

\ifnum\pdfoutput<1\else\def\svgmode{pdf}\fi\fi

\newcommand{\includesvg}[2][]{%
\def\tempa{#1}\def\tempb{}%
\ifx\tempa\tempb\else\let\svgwidth\tempa\fi
\executeiffilenewer{\svgpath#2.svg}{\svgpath#2.\svgmode}%
{inkscape -z -D --file=\svgpath#2.svg --export-latex %
--export-\svgmode=\svgpath#2.\svgmode}%
\input{\svgpath#2.\svgmode_tex}%
}

\graphicspath{{./figures/}}
\def\svgpath{figures/} 

\newcommand{\Z}{\mathbb{Z}}

\newcommand{\R}{\mathbb{R}}

\newcommand{\inv}{^{-1}}
\newcommand{\relint}[1]{\overset{\circ}{#1}}
\newcommand{\intForm}[2]{\langle #1,  #2\rangle}
\newcommand{\card}[1]{\#( #1 )}

\newcommand{\cut}{\backslash\!\!\backslash}

\definecolor{definecolor}{rgb}{0,0.1,0.55}
\def\define#1{\textbf{\textcolor{definecolor}{#1}}}

\title{\Large Computing the second and third systoles of a combinatorial surface}
\author[1]{Matthijs Ebbens}
\author[2]{Francis Lazarus\thanks{This author is partially supported by the French ANR projects MINMAX (ANR-19-CE40-0014) and the LabEx PERSYVAL-Lab (ANR-11-LABX-0025-01) funded by the French program Investissement d’avenir.}}

\affil[1]{University of Cologne, Germany}
\affil[2]{G-SCOP/Institut Fourier, CNRS, Université Grenoble Alpes, France}
\begin{document}
\maketitle

\begin{abstract}
  Given a weighted, undirected graph $G$ cellularly embedded on a topological surface $S$, we describe algorithms to compute the second shortest and third shortest closed walks of $G$ that are neither homotopically trivial in $S$ nor homotopic to the shortest non-trivial closed walk or to each other. Our algorithms run in $O(n^2\log n)$ time for the second shortest walk and in $O(n^3)$ time for the third shortest walk. We also show how to reduce the running time for the second shortest homotopically non-trivial closed walk to $O(n\log n)$ when both the genus and the number of boundaries are fixed. 

  Our algorithms rely on a careful analysis of the configurations of the first three shortest homotopically non-trivial curves in $S$. As an intermediate step, we also describe how to compute a shortest essential arc between \emph{one} pair of vertices or between \emph{all} pairs of vertices of a given boundary component of $S$ in $O(n^2)$ time or $O(n^3)$ time, respectively. 
\end{abstract}

\section{Introduction}

Computing a shortest homotopically non-trivial cycle of a combinatorial surface is one of the standard problems in computational topology, for which many algorithms have been developed~\cite{thomassen1990,erickson2004,k-csntc-06,cabello2007,cabellochambers2007,cabello2012,e-cocb-12,cce-msspe-13}. Instead of restricting to just the shortest homotopically non-trivial cycle, one can also consider the \emph{second} shortest not homotopic to the first one, the \emph{third} shortest not homotopic to the first two, and so on. More generally, to every cycle $c$ of the surface we can associate the length of the shortest cycle freely homotopic to $c$. By definition, this length only depends on the homotopy class of $c$. The ordered sequence of lengths of all free homotopy classes of cycles is called the \define{length spectrum} of the combinatorial surface, while the mapping from free homotopy classes to their lengths is called the \emph{marked} length spectrum. In the literature, a cycle corresponding to the first value of the length spectrum is usually referred to as the \define{systole} (which we also call \define{first systole)}, so in the same way we will refer to a cycle corresponding to the $i$-th value of the length spectrum as the \define{$i$-th systole}. 

The notion of length spectrum is well studied in the fields of hyperbolic and Riemannian geometry~\cite{o-smlsc-90,b-gscrs-92,p-islsq-18}. It is a remarkable result that the marked length spectrum of a non-positively curved surface completely determines the geometry of the surface~\cite{o-smlsc-90}. Even though a similar statement is not true for the unmarked length spectrum~\cite{vigneras1980}, it still provides a lot of information about the geometry of the surface~\cite{mckean1972,wolpert1979}.

Recently, a first effort has been made to study the length spectrum from the discrete point of view of combinatorial surfaces~\cite{ebbens2022,delecroix2024}. In particular, algorithms have been developed for computing, for any given positive integer $k$, the first $k$ values of the length spectrum of a weighted graph on the torus, as well as algorithms that decide equality of (un)marked length spectra of two given graphs on the torus. However, the approaches used in these papers heavily rely on the equivalence of homotopy and homology on the torus, which is not the case for surfaces of higher genus. 
The aim of this paper is to extend this work to surfaces of higher genus. Our main results are algorithms that compute a second and third systole.

\begin{theorem}\label{thm:secondsystole}
	Let $G$ be a weighted graph of complexity $n$ cellularly embedded on a surface $S$ of genus $g$ with $b$ boundary components. A second systole of $G$ can be computed in $O(n^2\log n)$ time or in $O((b^2+g^3)n\log n)$ time.
\end{theorem}

\begin{theorem}\label{thm:thirdsystole}
	Let $G$ be a weighted graph of complexity $n$ cellularly embedded on a surface $S$ of genus $g$. A third systole of $G$ can be computed in $O(n^3)$ time. 
\end{theorem}

We refer to Section~\ref{sec:discussion} for a discussion on whether our methods can be generalized to the computation of more values of the length spectrum. 

Our algorithms are based on a careful analysis of the possible configurations of the first three systoles, in particular in relation to their crossings. Our analysis can be summarized as the following theorem, which is a combination of Lemmas~\ref{lem:secondsystole},~\ref{lem:l3-simple},~\ref{lem:l3-l1-atmostonce},~\ref{lem:l3-l2-twice} and~\ref{lem:l3-l2-exactlytwice}.

\begin{theorem}\label{thm:configurations}
    The first systole $\ell_1$, second systole $\ell_2$ and third systole $\ell_3$ may be chosen in such a way that the following properties hold:
    \begin{enumerate}
        \item $\ell_1$ is simple,
        \item $\ell_2$ is simple or traverses $\ell_1$ twice,
        \item $\ell_3$ is simple or traverses $\ell_1$ twice or thrice,
        \item $\ell_2$ intersects $\ell_1$ at most once,
        \item $\ell_3$ intersects $\ell_1$ at most once,
        \item $\ell_3$ intersects $\ell_2$ at most twice,
        \item if $\ell_3$ is simple and intersects $\ell_2$ exactly twice, then $\ell_1$ is disjoint from $\ell_2$ and $\ell_3$ and is contained in a genus one subsurface bounded by either $\ell_2$ or $\ell_3$.
    \end{enumerate}
\end{theorem}

As a tool in our algorithm for computing the third systole, we also use algorithms that compute a shortest essential arc between one pair of vertices or between all pairs of vertices on a boundary component of the surface; we believe these algorithms are of independent interest.
The following theorem combines the results from Theorems~\ref{thm:shortest-arc} and~\ref{thm:shortest-arcs}. 

\begin{theorem}
    The shortest essential arc between a given pair of vertices on a boundary component of a combinatorial surface of complexity $n$ can be computed in $O(n^2)$ time. Moreover, after $O(n^3)$ time preprocessing, a shortest essential arc between a given pair of vertices on the boundary component can be computed in $O(n)$ time. As a consequence, shortest essential arcs between all pairs of vertices of the boundary component can be computed in $O(n^3)$ time.
\end{theorem}

\paragraph*{Organization of paper.}

We start by discussing some preliminaries in Section~\ref{sec:preliminaries}. In Section~\ref{sec:shortestessentialarcs} we present our algorithms for computing shortest essential arcs between one pair of vertices, or between all pairs of vertices of a given boundary curve. In Section~\ref{sec:lemmas} we give some basic lemmas that will prove useful in our analysis of the number of crossings between the first, second and third systoles, which follows in Section~\ref{sec:propertiesofsecondandthirdsystoles}. Our algorithms for computing the second and third systoles are given in Section~\ref{sec:computingsecondthirdsystole}. Finally, we discuss whether our methods can be generalized to compute further values of the length spectrum in Section~\ref{sec:discussion}.

\vspace{2mm}
\paragraph*{Acknowledgements.}

We would like to thank the referees for their helpful comments.

\section{Preliminaries}\label{sec:preliminaries}
\paragraph{Combinatorial curves and surfaces.}  We allow graphs to have loops and multiple edges. We view each edge of a graph as a pair of opposite half-edges, each one incident to an endpoint of that edge, called its \define{origin}. A half-edge can also be seen as an oriented edge directed away from the half-edge origin. We denote by $h\inv$ the opposite of a half-edge $h$. Let $S$ be a compact oriented surface without boundary. A cellular embedding of a graph $G$ in $S$ is a topological embedding of $G$ in $S$ whose complement in $S$ is a finite collection of connected sets each homeomorphic to an open disk. Such disks are the faces of the embedding. Up to a homeomorphism of $S$, this embedding is determined by a rotation system that records for every vertex of the graph the direct cyclic order of its incident half-edges with respect  to the orientation of $S$. Choosing a subset of faces and declaring them as perforated, we obtain the notion of \define{combinatorial surface} without or with boundary according to whether this subset is empty or not. $G$ is referred to as the graph of the (combinatorial) surface. Its boundary edges and  vertices are the ones incident to perforated faces. We shall always assume that $G$ has no isolated vertex. Note that they would correspond to sphere components of $S$.
We endow every combinatorial surface with a metric by choosing a positive weight function over the edges.
We denote the combinatorial surface with graph $G$ embedded on a surface $S$ by $(S, G)$. The edge weight function and the subset of perforated faces are implicit in this notation, and we may speak of a \emph{weighted} combinatorial surface to recall that its edges are weighted. The \define{complexity} of $(S, G)$ is the complexity of its graph, i.e. its total number of edges and vertices.

A \define{combinatorial curve} in a combinatorial surface is just a walk in its graph, that is a sequence of half-edges such that the origin of a half-edge coincides with the origin of the opposite of the previous half-edge in the sequence. A curve is closed when its endpoints coincide. An \define{arc} is a  walk whose endpoints are boundary vertices and whose incident perforated faces are well identified (hence, a choice must be made when an endpoint is incident to several perforated faces). The length $|c|$ of a combinatorial curve $c$ is the sum of the weights of its half-edges, counted with multiplicity, where the half-edges are given the weights of the corresponding edges.

\paragraph{Combinatorial crossings.} In order to define the notion of intersection between combinatorial curves unambiguously, we augment every set $\mathcal{C}$ of such curves  with a \define{combinatorial perturbation} as introduced in~\cite[Sec 5.1]{cl-oslos-05} (see also~\cite[Sec. 6]{lpvv-ccpso-01} and~\cite[Sec. 7]{dl-cginc-19}). This is the data for every half-edge of $G$ of a left-to-right order over all occurrences of that half-edge or its opposite in $\mathcal{C}$. We require that the left-to-right order of a half-edge is the reverse of the left-to-right order of its opposite half-edge. By concatenating the orders of all the half-edges with the same origin $v$, in the cyclic order given by the rotation system, we obtain a cyclic order at $v$ over all occurrences of these half-edges or their opposite in $\mathcal{C}$. For every perforated face $f$ incident with $v$, we also insert a dummy half-edge $h_{v,f}$ in this cyclic order in-between the two half-edges bounding $f$ (a face could be incident several times to $v$, in which case we insert dummy half-edges in-between each pair of half-edges bounding $f$ and sharing a common occurrence of $v$).

If $h$ is a half-edge occurrence in a curve $\gamma\in \mathcal{C}$, we denote by $h_{\text{prec}}$ the half-edge preceding $h$ along $\gamma$. This definition assumes that $h$ has a preceding half-edge along $\gamma$, or equivalently that the origin of $h$ is an interior vertex occurrence of $\gamma$. If $\gamma$ is an arc, then the origin $v$ of $h$ might be a boundary vertex associated to a perforated face $f$, in which case we set $h_{\text{prec}}=h_{v,f}$.

Given two distinct half-edge occurrences $h, k$ in $\mathcal{C}$ whose respective origins $v_1,v_2$ are distinct occurrences of a same vertex, we say that the curves in $\mathcal{C}$ have a \define{crossing at $(v_1,v_2)$} or \define{cross at $(v_1,v_2)$} if the four half-edge occurrences $h, k, h_{\text{prec}}\inv, k_{\text{prec}}\inv$ appear in this cyclic order or its opposite in the cyclic order at $v$. A curve is \define{weakly simple}\footnote{Weakly simple is called \emph{essentially simple} in~\cite[Th. 3.3]{ew-csec-10}.} if it admits a combinatorial perturbation without self-crossing. Similarly, two curves are \define{weakly disjoint} if they have a combinatorial perturbation without a crossing involving both curves. 

\paragraph{Duality and cross-metric surfaces.}
There is an alternative way of dealing with curves on combinatorial surfaces based on the dual notion of cross-metric surface~\cite{de2010}. The \define{cross-metric surface} associated to a combinatorial surface $(S,G)$ is an embedding of the dual graph $G^*$ on the same surface $S$. The graph $G^*$ has one dual vertex $f^*$  in every face $f$ of $G$ and a dual edge $e^*$ for every edge $e$ of $G$ connecting vertices dual to the faces incident to $e$. The duality actually extends to the half-edges of $e=(h,h\inv)$ by associating $h$ with the half-edge of $G^*$ whose  origin is dual to the face lying to the right of $h$, and similarly for $h\inv$. We identify as perforations the dual vertices corresponding to perforated faces. The reverse of this duality associates every cross metric surface with a combinatorial surface. In this reverse duality perforated vertices of $G^*$ maps to perforated faces of $G$ and the right face should be replaced by the left one to define half-edge duality.

The weight function over the edges of a weighted combinatorial surface transfers to a weight function over the dual edges of the corresponding cross-metric surface. A \define{normal curve} on a cross-metric surface is a continuous map from a circle or a line segment to $S$ that intersects $G^*$ transversely and away from its vertices, except for the extremities of arcs that must coincide with perforated vertices. The length $|\gamma|$ of a normal curve $\gamma$ is the sum of the weights of the edges crossed by this curve, counted with multiplicity. Note that normal curves can always be perturbed to isotopic curves with a finite number of crossings away from $G^*$ while preserving the ordered sequence of edges of $G^*$ crossed by each of the curves. Moreover, if two normal curves touch but do not cross, they can be perturbed to disjoint curves.

We associate with each crossing of a normal curve $\gamma$ with an edge $e^*$ of $G^*$ the half-edge of $e$ whose origin is dual to the face crossed by $\gamma$ just before the crossing. This way, 
the sequence of crossings of $\gamma$ with the edges of $G^*$ gives rise to a combinatorial curve $c_\gamma$ defined by the corresponding sequence of half-edges. $c_\gamma$ comes with a combinatorial perturbation defined for each half-edge $h$ of an edge $e$ by the ordering of the crossings of $\gamma$ along the edge $e^*$ oriented according to  the dual of $h$. Conversely, a combinatorial perturbation of a combinatorial curve $c$ can be realized as a normal curve $\gamma_c$ in a neighborhood of the support of $c$ in $S$ so that $c_{\gamma_c}=c$ comes with the same combinatorial perturbation.
This duality extends directly to set of curves and preserves their lengths, i.e. $|c|=|\gamma_c|$.

A normal curve intersects the interior of each dual face into a collection of arcs, that we call \define{strands}, with limit points in $G^*$. 
There are many ways to realize a combinatorial perturbation as a set of normal curves. In particular, the strands of the normal curves may or may not have self-crossings and cross pairwise more than once. To avoid this situation, we will always assume that the strands of the normal curves are in minimal configuration inside each dual face. This ensures that the number of combinatorial \mbox{(self-)cross}ings of the combinatorial perturbation is the same as the number of \mbox{(self-)cross}ings of its realization as normal curves~\cite[Lem. 27]{dl-cginc-19}. Even so, there might be minimally crossing realizations that are not isotopic inside each dual face. However, it can be proven that they differ by a sequence of 3-3 moves inside each face~\cite{gs-mcmcr-97,p-cails-02}. 

We will freely use the combinatorial and the dual cross-metric formalisms, depending on which one best suits the presentation.

\paragraph{Cutting a surface along a set of curves.}
Consider a combinatorial surface $\Sigma=(S,G)$, possibly with boundary, and a combinatorial perturbation $\Pi$ of a set $\mathcal{C}$ of combinatorial curves. We define the \define{cutting} of $\Sigma$ along $\mathcal{C}$ as a new combinatorial surface, which we denote by $\Sigma\cut \Pi$, or $\Sigma \cut \mathcal{C}$ when the combinorial perturbation is implicit. The precise definition of this operation, which is somewhat cumbersome, can be described more simply in the dual setting of  cross-metric surfaces and normal curves. Let $\Gamma$ be a realization  of $\Pi$, where the normal curves in $\Gamma$ realize the curves in $\mathcal{C}$. We cut $S$ through $\Gamma$, meaning that we take the completion (with respect to any PL or Riemannian metric) of the complement of $\Gamma$ in $S$. Denote by  $S'$ the resulting surface with boundary. $\Gamma$ also cuts the dual graph $G^*$ to give a graph $H^*$ embedded in $S'$.
We then contract every boundary component $\delta$ of $S'$. Such a contraction identifies all the intersections of $H^*$ with $\delta$ to a single point that we mark as a new perforation (in addition to the perforations of $G^*$, if any). The contractions of all these boundaries result in a closed surface $S''$, possibly not connected, and a graph $K^*$ with (vertex) perforations that is cellularly embedded on $S''$. It thus provides a cross-metric structure for $S''$ where the pieces of an edge of $G^*$ in $K^*$ are each assigned the same weight as that edge. We throw away the components of $S''$ intersecting $K^*$ in a single (necessarily perforated) vertex.
$\Sigma\cut \Pi$ can now be defined as  the dual combinatorial surface. The weight assignment ensures that any curve in $\Sigma\cut \Pi$ identifies to a curve of $\Sigma$ with the same length. It is not difficult to see that although $\Gamma$ is defined up to 3-3 moves inside faces, the corresponding constructions of $S''$ and $K^*$ are invariant under such moves, so that $\Sigma\cut \Pi$ is well-defined.

\paragraph{Shortest path trees and tree loops.}
For any vertex $s$ of a connected graph $G$, we denote by $T_s$ a tree of shortest paths from the source $s$ to every other vertex of $G$. In general, unless we assume that every two vertices are joined by a unique shortest path, there might be several shortest path trees with the same source $s$. We thus consider $T_s$ as one of them, say as computed by running Dijkstra's algorithm. For every vertex $u$ and every edge $uv$,  we then denote by $T_s(u)$ the shortest path from $s$ to $u$ in $T_s$ and by $T_s(uv) := T_s(u)\cdot uv\cdot T_s(v)\inv$ the tree loop resulting from the concatenation of the edge $uv$ with the tree paths between its endpoints and $s$. Suppose that $G$ is cellularly embedded in $S$, and let $G^*$ be the dual graph with respect to this embedding. Following~\cite{ew-csec-10}, we define the \define{cut locus} of $s$ as the subgraph obtained from $G^*$ by removing all the edges dual to the edges of $T_s$. The cut locus may be seen as the adjacency graph of the faces of the combinatorial surface $(S,G)$ cut along $T_s$. It easily follows that a loop $T_s(uv)$ is contractible in $S$ if and only if either $uv$ is an edge of $T_s$, or removing the dual edge from the cut locus leaves at least one tree component~\cite[Lem. 4.1]{ccl-fctpe-11}. If $(S,G)$ has perforated faces, we also require that this tree component does not contain any dual vertex perforation. By removing the leaves of the cut locus iteratively, checking that it is not marked as a perforation, we obtain all the dual edges $uv^*$ such that $T_s(uv)$ is contractible. This clearly takes linear time.

\section{Shortest essential cycles and arcs}\label{sec:shortestessentialarcs}
A closed curve in $S$ is \define{essential} if it simple and neither homotopic to a boundary component nor to a point. On a combinatorial surface, where curves lie in the graph (1-skeleton) of the surface, the property of being simple should be replaced by \emph{weakly simple}. In this framework, Erickson and Worah~\cite[th. 3.3]{ew-csec-10} proved that the shortest essential cycle can be decomposed into an alternating concatenation of an edge, a shortest path, an edge and a shortest path. This decomposition allows them to
compute the shortest essential weakly simple cycle in $O(n^2\log n)$ time, or in $O(n\log n)$ time on a combinatorial surface of complexity $n$, when the genus and number of boundary components are fixed. Note that Erickson and Worah need to assume that shortest paths are unique and that every edge is a shortest path between its endpoints for their algorithm to work correctly.
The first of these assumptions can be circumvented using the notion of \emph{lex shortest paths} introduced by Hartvigsen and Mardon~\cite{hartvigsen1994}: given vertices $u,v\in V$ there exists a unique path $P$ (called the lex shortest path), such that for any other path $P'$ from $u$ to $v$ either $w(P)<w(P')$, or $w(P)=w(P')$ and $\ell(P)<\ell(P')$ (where $\ell(P)$ denotes the number of edges in $P$), or $w(P)=w(P')$ and $\ell(P)=\ell(P')$ and $\min(V(P)\setminus V(P'))<\min(V(P')\setminus V(P))$. A collection of lex shortest path trees rooted at every vertex can be computed in $O(n^2\log n)$ time~\cite[Propositions 4.10 and 4.11]{hartvigsen1994}. Note that Proposition 4.11 is stated for planar graphs, but there is no need for this assumption. The paths in the lex shortest path trees can then be used as unique shortest paths. Therefore, we may henceforth assume without loss of generality that after $O(n^2\log n)$ preprocessing time shortest paths are unique\footnote{See~\cite[Sec. 6]{cce-msspe-13} and~\cite[Sec. 3]{efl-hmcpf-18} for more details on how to enforce uniqueness of shortest paths}.

In addition to the computation of shortest essential weakly simple cycles, our toolbox for computing the second and third systole requires an algorithm to compute shortest essential arcs. For arcs, we use a relaxed definition of essential: an arc with endpoints on the boundary of $S$ is \define{inessential} if it can be homotoped to the boundary of $S$. In particular, it may wind several times around a boundary component.
An arc is \define{essential} if it is not inessential.
Of course, any arc $\alpha$ with endpoints on distinct boundary components is essential.
We insist that an essential arc need not be (weakly) simple. In Lemma~\ref{lem:path-edge-path} below, we state a decomposition of shortest essential arcs similar to the decomposition of shortest essential cycles~\cite[Th. 3.3]{ew-csec-10}. Namely, a shortest essential arc is the concatenation of two shortest paths, possibly with an edge in-between. We leverage this decomposition to provide algorithms for computing an essential arc between \emph{a given} pair of endpoints (Theorem~\ref{thm:shortest-arc}), or for computing shortest essential arcs between \emph{every} pair of endpoints on a boundary curve (Theorem~\ref{thm:shortest-arcs}).

\subsection{Computing shortest essential arcs.}
Here, we present simple algorithms for computing shortest essential arcs in the combinatorial framework. We only consider arcs with endpoints on a same boundary component, since otherwise the problem is trivially solved by computing shortest paths between the endpoints. As opposed to the work of Erickson and Worah~\cite{ew-csec-10} for computing shortest essential cycles,  \emph{we do not need to assume that shortest paths are unique, nor that edges are shortest paths between their endpoints}, even if this last condition is easy to enforce.

In the following lemmas, we keep the classical notation $p\cdot q$ for the concatenation of two paths $p,q$. We write $p\sim q$ to indicate that $p$ and $q$ are homotopic with fixed endpoints. We first state a decomposition for shortest essential arcs analogous to the decomposition of shortest essential cycles given in~\cite{ew-csec-10}.
\begin{lemma}\label{lem:path-edge-path}
  Let $x,y$ be two distinct vertices on a boundary component $\delta$ of a combinatorial surface. Any shortest essential arc between $x$ and $y$ is composed of a shortest path from $x$ to some vertex $u$, an edge $uv$, and a shortest path from $v$ to $y$.
\end{lemma}
\begin{proof}
  Let $\alpha$ be a shortest essential arc between $x$ and $y$. By a vertex \emph{along} $\alpha$ we mean an occurrence of that vertex in $\alpha$. Hence, if $\alpha$ goes several times through a same vertex we may have two distinct occurrences $u$ and $v$ of that same vertex along $\alpha$ and consider the (non-trivial) subarc of $\alpha$ from $u$ to $v$. More generally,
if $u,v$ are vertices along $\alpha$, with a little abuse of notation,  we denote by $\alpha_{uv} = \alpha_{vu}\inv$ the subarc of $\alpha$ from $u$ to $v$. We claim that for any vertex $u$ along $\alpha$, either $\alpha_{xu}$ or $\alpha_{uy}$ is a shortest path. Suppose otherwise and let $\sigma_{xu}, \sigma_{uy}$ be shortest paths from $x$ to $u$ and from $u$ to $y$, respectively. Since $\alpha_{xu}\cdot \sigma_{uy}$ is shorter than $\alpha$, it must be inessential. In other words $\alpha_{xu}\cdot \sigma_{uy}\sim \delta_{xy}\cdot \delta^k$ for some integer $k$. Analogously, we have $\sigma_{xu}\cdot \sigma_{uy}\sim \delta_{xy}\cdot \delta^\ell$ and $\sigma_{xu}\cdot \alpha_{uy}\sim \delta_{xy}\cdot \delta^m$ for some integers $\ell,m$. It ensues that
  \[\alpha = \alpha_{xu}\cdot \alpha_{uy}\sim (\alpha_{xu}\cdot \sigma_{uy})\cdot (\sigma_{xu}\cdot \sigma_{uy})\inv \cdot (\sigma_{xu}\cdot \alpha_{uy}) \sim \delta_{xy}\cdot  \delta^{k+m-\ell}.
    \]
    This however contradicts the hypothesis that $\alpha$ is essential and concludes the proof of the claim.

    We now consider the last vertex $u$ along $\alpha$ such that $\alpha_{xu}$ is a shortest path. If $u=y$ then $\alpha$ is a shortest path and the lemma is trivially true. Otherwise let $v$ be the vertex next to $u$ along $\alpha$. Since $\alpha_{xv}=\alpha_{xu}\cdot uv$ is not a shortest path, it follows by the above claim that $\alpha_{vy}$ is a shortest path. We have thus obtained the required decomposition.
 \end{proof}

Choose a shortest path between every pair of vertices. We claim that we can find a shortest essential arc among paths that decompose into two of these shortest paths or into a shortest path, an edge, and a shortest path. Note that a concatenation of two shortest paths also decomposes into a shortest path, an edge, and a shortest path. However, the non-uniqueness of shortest paths makes the set of paths of the first type a priori not included in the set of paths of the second type.

 \begin{lemma}\label{lem:shortest-trick}
   Let $x,y$ be two distinct vertices on a boundary component $\delta$ of a combinatorial surface $\Sigma$. Suppose that we have for each vertex $u$ a shortest path $\sigma_{xu}$  between $x$ and $u$ and a shortest path\footnote{Here, if $v$ is a vertex along $\sigma_{xu}$ or $\sigma_{uy}$, we do not assume that $\sigma_{xv}$ is a subpath of $\sigma_{xu}$, nor that $\sigma_{vy}$ is a subpath of $\sigma_{uy}$.} $\sigma_{uy}$ between $u$ and $y$. Then, the set of paths of the form
   \begin{itemize}
     \item $\sigma_{xu}\cdot \sigma_{uy}$, for $u$ a vertex of $\Sigma$, or
   \item $\sigma_{xu}\cdot uv\cdot \sigma_{vy}$, for $uv$ an (oriented) edge of $\Sigma$
   \end{itemize}
    contains a shortest essential arc between $x$ and $y$.
\end{lemma}
\begin{proof}
  Let $\alpha$ be a shortest essential arc between $x$ and $y$. Denote by $x=u_1,u_2,\dots, u_k=y$ the vertices along $\alpha$. We have
  \[\alpha \sim (u_1u_2\cdot \sigma_{u_2y})\cdot (\sigma_{xu_2}\cdot \sigma_{u_2y})\inv \cdot (\sigma_{xu_2}\cdot u_2u_3\cdot \sigma_{u_3y})\cdot (\sigma_{xu_3}\cdot \sigma_{u_3y})\inv\cdots (\sigma_{xu_{k-1}}\cdot u_{k-1}u_k).
  \]
Hence, if all the paths $\gamma_{i}:=\sigma_{xu_{i}}\cdot u_{i}u_{i+1}\cdot \sigma_{u_{i+1}y}$ and $\lambda_{i}:=\sigma_{xu_{i}}\cdot \sigma_{u_{i}y}$ were inessential, then so would be $\alpha$, in contradiction with the hypothesis. It follows that one of the $\gamma_i$'s  or  $\lambda_i$'s is essential. Since such a path is no longer than $\alpha$, it must be a shortest essential arc.
\end{proof}
The previous lemma leads to a quadratic time algorithm for computing a shortest essential arc between two given vertices.
\begin{theorem}\label{thm:shortest-arc}
  Let $x,y$ be two distinct vertices on a boundary component $\delta$ of a combinatorial surface $\Sigma$. A shortest essential arc between $x$ and $y$ can be computed in $O(n^2)$ time, where $n=|\Sigma|$ is the total number of edges and vertices of $\Sigma$.
\end{theorem}
\begin{proof}
  Using Dijkstra's algorithm we compute shortest path trees $T_x$ and $T_y$ with respective roots $x$ and $y$ in $O(n\log n)$ time. We store at each vertex of $T_x$ its distance to $x$ and do similarly for $T_y$. We consider for every vertex $u$ and for every edge $uv$, the arcs $\gamma_u := T_x(u)\cdot T_y(u)\inv$ and $\gamma_{uv} := T_x(u)\cdot uv\cdot T_y(v)\inv$. We set $\Gamma_V = \{\gamma_u\mid u\in V\}$ and $\Gamma_E=\{\gamma_e\mid e\in \vec{E}\}$, where $V$ is the set of vertices and $\vec{E}$ is the set of oriented edges (i.e. half-edges) of $\Sigma$.

  We can compute the length of each arc in $\Gamma_V\cup \Gamma_E$ in constant time thanks to the distances stored in the shortest path trees. We argue below that we can test in $O(n)$ time whether such an arc is essential or not. It then follows from Lemma~\ref{lem:shortest-trick} that we can select in $O(n^2)$ time a shortest essential arc between $x$ and $y$ among the $O(n)$ arcs in $\Gamma_V\cup \Gamma_E$. 

  It remains to argue that we can test essentiality in linear time. We first claim that a path $\gamma$ from $x$ to $y$ is inessential if and only if the commutator $[\gamma\cdot\delta_{yx},\delta] = (\gamma\cdot\delta_{yx})\cdot \delta\cdot (\gamma\cdot\delta_{yx})\inv \cdot \delta\inv$ is contractible in $S$. The direct implication is clear, since $\gamma$ being inessential implies $\gamma\sim \delta_{xy}\cdot \delta^k\sim\delta^{k}\cdot \delta_{xy}$ for some $k$. For the reverse implication, we use the well-known fact~\cite[p. 213]{r-ajccs-62}~\cite[Lem. 9.2.6]{b-gscrs-92} that two elements in the fundamental group of a surface -- distinct from the torus and the Klein bottle -- commute if and only if they have a common primitive root\footnote{Here, we only need this result for the easier case of free groups, since we are only considering surfaces with boundary.}. Since $\delta$ is primitive, the triviality of the commutator implies that $\gamma\cdot\delta_{yx}$ is homotopic to a power of $\delta$, or equivalently that $\gamma\sim \delta_{xy}\cdot \delta^k$, for some integer $k$. In other words, $\gamma$ is inessential.
  
Thanks to the claim,  we can detect in linear time whether $\gamma$  is inessential. Indeed, testing if a curve with $k$ edges is contractible can be done in $O(n+k)$ time~\cite{lr-hts-12,ew-tcsr-13} and each curve $\gamma_u$, $\gamma_{uv}$, and $\delta$ has linear size, as well as their commutators.
\end{proof}

 Even though essential arcs need not be simple in our definition, we show below that for every pair of boundary vertices there is a shortest essential arc between them that is weakly simple. We first recall how to resolve a crossing. Suppose that a path $\alpha$ self-crosses at some occurrences $u$ and $v$ of the same vertex, with $u$ occurring before $v$. Then the path $\alpha':=\alpha_{xu}\cdot \alpha_{uv}\inv\cdot\alpha_{vy}$ has one crossing fewer than $\alpha$. We say that $\alpha'$ is obtained from $\alpha$ by \emph{resolving}, or \emph{smoothing} its crossing at $u=v$.
 \begin{lemma}\label{lem:weakly-simple}
   Let $\alpha$ be a shortest essential arc between two vertices $x,y$ on a boundary component $\delta$ of a combinatorial surface $\Sigma$. Then,  the arc obtained by resolving iteratively all the crossings of $\alpha$ is a weakly simple shortest essential  arc.
 \end{lemma}
 \begin{proof}
   We may assume that the endpoints $x,y$ of $\alpha$ lie on a same boundary component $\delta$ of $S$, since otherwise $\alpha$ is a shortest path, hence is simple.
   Suppose that $\alpha$ has at least one crossing at $u=v$.
   We show below that the path $\alpha'$ obtained from $\alpha$ by resolving its crossing at $u=v$ is essential. Since $\alpha'$ has the same length as $\alpha$, this will prove the lemma by induction on the number of crossings of $\alpha$.
   Now, suppose for a contradiction that $\alpha'$ is inessential, i.e. $\alpha_{xu}\cdot \alpha_{uv}\inv\cdot\alpha_{vy}\sim \delta_{xy}\cdot \delta^k$ for some integer $k$. Since $\alpha_{xu}\cdot \alpha_{vy}$ is shorter than $\alpha$, we also have $\alpha_{xu}\cdot \alpha_{vy}\sim \delta_{xy}\cdot \delta^\ell$ for some $\ell$. We infer that
   \[ \alpha = (\alpha_{xu}\cdot \alpha_{uv})\cdot\alpha_{vy}\sim \alpha_{xu}\cdot \alpha_{vy} \cdot (\alpha_{xu}\cdot \alpha_{uv}\inv\cdot\alpha_{vy})\inv\cdot (\alpha_{xu}\cdot\alpha_{vy})\sim \delta_{xy}\cdot \delta^{2\ell-k},
   \]
   which contradicts the fact that $\alpha$ is essential. It follows that $\alpha'$ is essential.
 \end{proof}

 By Lemma~\ref{lem:weakly-simple}, we can extract a weakly simple essential arc from the shortest essential arc computed as in Theorem~\ref{thm:shortest-arc} by resolving its crossings. We argue below that this can be done in linear time.
\begin{lemma}\label{lem:compute-weakly-simple}
  Given a decomposition of a shortest essential arc $\alpha$ into a shortest path $\sigma_{xu}$, an edge $uv$ and a shortest path $\mu_{vy}$, we can compute a weakly simple shortest essential arc using the same edges as $\alpha$ in time proportional to its number of edges.
\end{lemma}
\begin{proof}
  Since $\sigma_{xu}$ is a simple path and since $\alpha$ is a shortest essential arc, it is easily seen that edge $uv$ cannot appear in $\sigma_{xu}$. Similarly, $uv$ cannot appear in $\mu_{vy}$. We may thus split $uv$, introducing a new vertex $z$ in the middle, and consider the \emph{simple} paths $\sigma_{xz}:= \sigma_{xu}\cdot uz$ and $\mu_{zy}:=zv\cdot\mu_{vy}$. Denote by $z=u_0, u_1,\dots, u_p=x$ the vertices along $\sigma_{xz}\inv$, and by $z=v_0, v_1,\dots, v_q=y$  the vertices along $\mu_{zy}$. Suppose that $\sigma_{zx}$ and $\mu_{zy}$ cross twice at ${u_i}={v_k}$ and ${u_j}={v_\ell}$. Without loss of generality, we assume $i<j$.

  We claim that $k<\ell$. By way of contradiction suppose that $k>\ell$. We have
  \[ \alpha = \sigma_{xz}\cdot\mu_{zy}\sim (\sigma_{xu_i}\cdot\mu_{v_k y})\cdot (\mu_{yv_k}\cdot \sigma_{u_i z}\cdot\mu_{zv_\ell}\cdot\sigma_{u_j x})\cdot (\sigma_{xu_j}\cdot\mu_{v_\ell y})
  \]
  Since $u_i$ appears before $u_j$ along $\sigma_{zx}$, while $v_\ell$ appears before $v_k$ along $\mu_{zy}$, each of the three arcs $\sigma_{xu_i}\cdot\mu_{v_k y}$, $\mu_{yv_k}\cdot \sigma_{u_i z}\cdot\mu_{zv_\ell}\cdot\sigma_{u_j x}$ and $\sigma_{xu_j}\cdot\mu_{v_\ell y}$ is strictly shorter than $\alpha$. It ensues that they are inessential. This however implies that $\alpha$ is also inessential, a contradiction.

  It follows from the claim that the sequence of crossings of $\sigma_{zx}$ and $\mu_{zy}$ appears in the same order along $\sigma_{zx}$ and $\mu_{zy}$. We can thus traverse $\sigma_{zx}$ and $\mu_{zy}$ in parallel to detect their crossings and smooth each crossing in the unique possible way leading to a connected curve. This clearly takes time proportional to the number of edges of $\alpha$.
\end{proof}

\subsection{Faster algorithm for all pairs shortest essential arcs.}
 
For the purpose of computing a third systole, we need to compute shortest essential arcs between all pair of vertices of a certain boundary curve $\delta$. Thanks to Theorem~\ref{thm:shortest-arc}, this can be done in $O(n^4)$ time since there are $O(n^2)$ pairs of vertices to consider. We show below how to reduce the computation time to $O(n^3)$. Our strategy is the following. For every vertex $u$, choose a shortest path tree $T_u$ with source $u$.
Consider two vertices $x,y$ on the boundary curve $\delta$, and recall the notations $\Gamma_V$ and $\Gamma_E$ from the proof of Theorem~\ref{thm:shortest-arc}.  Lemma~\ref{lem:shortest-trick} implies that we can find a shortest essential arc between $x$ and $y$ among the sets $\Gamma_V$ and $\Gamma_E$, each  containing $O(n)$ arcs. When $x$ and $y$ run over all the vertices of $\delta$ the union of these sets thus constitutes a family of $O(n^3)$ arcs. As we scan this family, we cannot afford to spend more than constant time per arc if we want to obtain a cubic time algorithm. In~\cite{ew-csec-10} an analogous situation occurs for the computation of a shortest essential cycle and the authors succeed to perform the computation required for each cycle in $O(1)$ time after precomputing some adequate data-structures. However, the correctness of their algorithm strongly relies on the uniqueness of shortest paths to detect weakly simple cycles. Although we do not assume uniqueness of shortest paths, a similar approach applies to our sets $\Gamma_V$ as they are in fact composed of weakly simple arcs by construction. This is not the case for the arcs in $\Gamma_E$, for which we adopt a different strategy: we decompose each such arc into a concatenation of an arc in  $\Gamma_V$ and a loop that only depends on two vertices. The number of such loops is therefore quadratic and we can now spend linear time at each of them to extract the desired information.

\begin{theorem}\label{thm:shortest-arcs}
  Let $\delta$ be a boundary curve of a combinatorial surface $(S,G)$ of complexity $n$. After $O(n^3)$ time preprocessing, we can compute for every pair $x,y$ of distinct vertices of $\delta$ a shortest weakly simple essential arc between $x$ and $y$ in $O(n)$ time. 
\end{theorem}
\begin{proof}
  We first compute a shortest path tree $T_v$ for every vertex $v$ not on $\delta$, and store the distance to $v$ at each vertex of the tree. This takes $O(n^2\log n)$ time in total using $O(n)$ executions of Dijkstra's algorithm. Next, for every vertex $x$ of $\delta$ and every edge $uv$, we form the loop $\lambda(x,uv):= T_u(x)\inv\cdot uv\cdot T_v(x)$ and check if it is homotopic (with basepoint $x$) to some power $\delta^k$ of $\delta$. We store the answer in a boolean array $\Lambda[x,uv]$. The edges of $\lambda(x,uv)$ can be listed in linear time from the shortest path trees. Testing whether $\lambda(x,uv)$ is homotopic to a power of $\delta$ amounts to check if the commutator $[\lambda(x,uv),\delta]$ is contractible, which takes linear time according to~\cite{lr-hts-12,ew-tcsr-13}. Denoting by $d$ the number of vertices of $\delta$, we conclude that $\Lambda$ has size at most $dn=O(n^2)$ and can be filled in $O(n^3)$ time.

  Let $z_0,z_1, \dots, z_{d-1}$ be the vertices of $\delta$.  We built for every vertex $u$ not on $\delta$, a boolean array $c_u$ of size $d$ so that $c_u[i]$ indicates whether $T_u(z_{i-1}z_{i}):= T_u(z_{i-1})\cdot z_{i-1}z_{i}\cdot T_u(z_{i})\inv$ is contractible or not. As was explained at the end of Section~\ref{sec:preliminaries}, we can fill $c_u$ in linear time. We also fill an integer array $m_u$ of size $d$ so that, if  $c_u[j] = \text{True}$, then
$m_u[j] = \min \{i\mid c_u[i+1]=c_u[i+2]=\dots=c_u[j] = \text{True}\}$.
Note that $m_u[j]$ might be negative and that the indices in the brackets are taken modulo $d$. If $c_u[j] = \text{False}$, we ignore $m_u[j]$.
By traversing $c_u$ at most twice, $m_u$ can be filled in linear time. Hence, computing the $c_u$'s and $m_u$'s for all $u$ takes time $O(n^2)$.

We now fix two distinct indices $0\le i < j < d$ and denote by
$\delta_{ij}=\delta_{ji}\inv$ and $\delta'_{ij}={\delta'_{ji}}\inv$ the two subarcs of $\delta=\delta_{ij}\cdot \delta'_{ji}$ with respective vertices $z_i, z_{i+1},\dots z_j$ and $z_j, z_{j+1},\dots z_i$, where all the indices are taken modulo $d$.
For every vertex $u$, let $\gamma(u) := T_u(z_i)\inv\cdot T_u(z_j)$. This is  a weakly simple arc since two paths with a common origin in a tree do not cross. We set $\Gamma_V^{ij}=\{\gamma(u)\mid u\in V\setminus \delta\}$.
  \begin{claim}\label{clm:gammaV}
    If $\Gamma_V^{ij}$ contains essential arcs, we can compute a shortest one among them in $O(n)$ time.
  \end{claim}
  To prove the claim, we first note that $\gamma(u)$ is inessential if and only if one of the two weakly simple cycles $\gamma(u)\cdot \delta_{ji}$ or $\gamma(u)\cdot \delta'_{ji}$ is contractible. Indeed, by definition, $\gamma(u)$ is inessential if $\gamma(u)\cdot \delta_{ji}$ is homotopic to a power of $\delta$. Since $\gamma(u)\cdot \delta_{ji}$ is weakly simple, the only possible powers are $\delta$,  $\delta\inv$, or the constant loop 1. (Recall that a curve homotopic to a simple curve must be either contractible or primitive.) In fact, we cannot have $\gamma(u)\cdot \delta_{ji}\sim \delta$ since it would imply $\gamma(u)\cdot \delta'_{ji}\sim \delta^2$, which is impossible as $\gamma(u)\cdot \delta'_{ji}$ is also weakly simple. It follows that $\gamma(u)$ is inessential if and only if 
  either $\gamma(u)\cdot \delta_{ji}$ or $\gamma(u)\cdot \delta'_{ji}$ is contractible. Now, if $\gamma(u)\cdot \delta_{ji}$ is contractible, then it bounds a disk (after infinitesimal perturbation). In addition to $T_u(z_i)$, this disk contains all the paths $T_u(z_k)$ for $k = i+1,\dots, j$, implying that the tree loops $T_u(z_{k-1}z_{k})$ are contractible. Hence, $\gamma(u)\cdot \delta_{ji}$ is contractible if and only if all the $T_u(z_{k-1}z_{k})$ are contractible, for $k = i+1,\dots, j$. This last condition is equivalent to the fact that $c_u[j]=\text{True}$ and that $i$ belongs to the interval $[m[j], j]:=(m[j], m[j]+1,\dots, j)$, where the indices are taken modulo $d$ in the tuple. We can thus decide in constant  time if $\gamma(u)\cdot \delta_{ji}$ is contractible.
  Similarly, we can decide in constant time if $\gamma(u)\cdot \delta'_{ji}$ is contractible by testing if $c_u[i]=\text{True}$ and if $j$ belongs to $[m[i], i]$. We can thus check for every $u$ if $\gamma(u)$ is essential in constant time. If not, we compute its length in constant time by adding the stored length of $T_u(z_i)$ and $T_u(z_j)$. A shortest essential arc $\gamma(u)$ in $\Gamma_V^{ij}$ is obtained for a $u$ that minimizes this length. Claim~\ref{clm:gammaV} follows.

  We then set $\gamma(uv) := T_{u}(z_i)\inv\cdot uv\cdot T_v(z_j)$ and $\Gamma_E^{ij}=\{\gamma(e)\mid e\in \vec{E}\}$.
  \begin{claim}\label{clm:gammaE}
    If the shortest essential arcs of $\Gamma_V^{ij}\cup \Gamma_E^{ij}$ are all contained in $\Gamma_E^{ij}$, then one among them can be computed in $O(n)$ time.
  \end{claim}
  Indeed, we must have $\gamma(u)\sim 1$ for any arc $\gamma(uv)$ in $\Gamma_E^{ij}$ that is a shortest essential arc since $\gamma(u)$ is shorter than $\gamma(uv)$. Now, we observe that
  \[\gamma(uv) \sim \gamma(u)\cdot \lambda(y,uv).
  \]
  It follows that $\lambda(y,uv)$ is not homotopic to a power of $\delta$. This last condition is recorded in $\Lambda[y,uv]$ and can thus be tested in constant time.
Conversely, if $\gamma(u)\sim 1$ and $\lambda(y,uv)$ is not homotopic to a power of $\delta$, then $\gamma(uv)$ is essential. It ensues that a shortest among such $\gamma(uv)$ is a shortest essential arc. Claim~\ref{clm:gammaE} easily follows.

Hence, for any two indices $0\le i < j < d$, Claims~\ref{clm:gammaV} and~\ref{clm:gammaE} allow us to compute a shortest essential arc with endpoints $z_i$ and $z_j$ in linear time. Using Lemma~\ref{lem:compute-weakly-simple}, we can transform this arc in linear time into a weakly simple arc. This ends the proof of Theorem~\ref{thm:shortest-arcs}.
\end{proof}

Erickson and Worah~\cite[Sec. 5]{ew-csec-10} propose even faster algorithms for computing shortest essential cycles when the genus and number of boundaries of $S$ are fixed. Since their algorithms take at least linear time, using a similar approach for computing shortest arcs between all pairs of vertices on $\delta$ would lead to a total running time at least $n$ times the number of pairs of vertices. In the worst case this would lead to an algorithm with time complexity at least cubic and would not improve Theorem~\ref{thm:shortest-arcs}.

\section{Curves and crossings}\label{sec:lemmas}

Here, we state some properties concerning curves and their crossings that will prove useful in our analysis of the number of crossings of the first three systoles. We do this in the framework of cross-metric surfaces. 
The properties essentially state that given a set of simple tight curves in minimal configuration, any other simple curve has a simple tight representative that intersects each curve in the set minimally. The proofs rely on classical exchange arguments based on Hass and Scott's bigon lemma~\cite[Lem. 3.1]{hs-ics-85}.

We first fix some notation conventions. Let $S$ be a compact orientable surface. In this section, we often consider curves on $S$ up to orientation. Hence, in order to avoid orientation considerations when concatenating two paths, say $p,q$, with the same endpoints we write $p \cup q$ rather than $p\cdot q$, $p\cdot q^{-1}$, $p^{-1}\cdot q$, or $p^{-1}\cdot q^{-1}$, whichever is meaningful. This only applies when $p$ and $q$ are not closed as otherwise the homotopy class of the concatenation $p \cup q$ is ambiguous. When they are closed, we thus keep the classical notation $p\cdot q$ for their concatenation. While the notation $p\sim q$ indicates that the paths $p$ and $q$ are homotopic with fixed endpoints, for two closed curves $\alpha, \beta$, the notation $\alpha\sim\beta$ refers to free homotopy, unless stated otherwise.

\paragraph{Algebraic vs. geometric intersection number.} The \define{algebraic intersection number} of two closed curves $\alpha, \beta$ on $S$ is denoted by $\intForm{\alpha}{\beta}$. It is well-known that this number only depends on the homology classes of $\alpha, \beta$. The \define{geometric intersection number} of $\alpha$ and $\beta$ is the minimal number of crossings among all pair of loops homotopic to $\alpha$ and $\beta$, respectively.  This is by definition a homotopy invariant, but in general, this is not a homology invariant.

In the next three lemmas we only consider curves that are normal with respect to the graph of the cross-metric surface. For conciseness, we thus drop the adjective normal. 
A curve is \define{tight} if it has minimal length (with respect to the cross-metric) in its homotopy class. The cardinality of a set $\cal S$ is denoted $\#\cal{S}$.
\begin{lemma}\label{lem:excess-intersection}
Let $\ell_1,\ell_2,\dots,\ell_k$ be simple, tight loops in minimal configuration, i.e. they pairwise cross minimally thus realizing their geometric intersection number. If $\ell$ is homotopic to a simple loop, then there exists a simple tight homotopic loop $\ell'$ that intersects each $\ell_i$ minimally. In other words the simple tight loops  $\ell_1,\ell_2,\dots,\ell_k, \ell'$ are in minimal configuration.
\end{lemma}
\begin{proof}
  Consider a tight loop homotopic to $\ell$. In order to make it simple and tight, use \cite[Th. 2.7]{hs-ics-85} for the existence of embedded 1 or 2-gon for a curve with excess self-intersection homotopic to a simple curve. This allows us to untie bigons (monogons cannot appear by tightness) until there are no more. Then to make $\ell$ cross minimally the $\ell_i$, use \cite[Lem. 3.1]{hs-ics-85} for the existence of an empty (innermost) embedded bigon and switch the $\ell$ part in each bigon. This leaves the $\ell_i$ unchanged and does not introduce self-crossings in $\ell$. Note that by tightness of the curves any bigon is bounded by two arcs of the same length, so that the above curve modifications indeed preserve tightness.
\end{proof}

\begin{lemma}\label{lem:contractible-case}
  Let $\ell, \ell'$ be simple, tight, non-homotopic loops on $S$ with at least two crossings cutting $\ell$ into paths $p,q$, and cutting $\ell'$ into paths $p', q'$. If $p$ is homotopic to $p'$ then there exists a simple tight loop $\ell''$ homotopic to $\ell'$ such that $\card{\ell \cap \ell''} < \card{\ell \cap \ell'}$.
\end{lemma}
\begin{proof}
  By exchanging $p'$ with $p$ in $\ell'$ and applying a small perturbation, we get a curve homotopic to $\ell'$, possibly with self-crossings, and at least one crossing with $\ell$ fewer. We then apply the previous lemma.
\end{proof}

\begin{lemma}\label{lem:homotopic-case}
  Let $\ell, \ell'$ be simple, tight, non-homotopic loops on $S$ with at least two crossings $u,v$. Let $u,v$ cut $\ell$ into paths $p,q$, and cut $\ell'$ into paths $p', q'$. If $p\cup p'$ is homotopic to $\ell$ or to $\ell'$, then there exists a simple tight loop $\ell''$ homotopic to $\ell'$ such that $\card{\ell \cap \ell''} < \card{\ell \cap \ell'}$.
\end{lemma}
\begin{proof}
  For an arc $a$, denote by $\overset{\circ}{a}$ its relative interior. We have
  \[ \card{\ell\cap \ell'} = \card{\relint{p}\cap \relint{p'}} + \card{\relint{p}\cap \relint{q'}} + \card{\relint{q}\cap \relint{p'}} + \card{\relint{q}\cap \relint{q'}} +2
  \]
  Set $\gamma = p\cup p'$. 
  If $\gamma \sim \ell$ then, after a small perturbation of $\gamma$, we obtain
  \[ \card{\gamma \cap \ell'} =\card{\relint{p}\cap \relint{p'}}  + \card{\relint{p}\cap \relint{q'}}  + \varepsilon
  \]
  with $\varepsilon = 0$ or 1 according to whether the ends of $\relint{p}$ are on the same side or on opposite sides of $\ell'$. In both cases, we have $\card{\ell\cap \ell'} > \card{\gamma \cap \ell'}$ showing that $\ell$ and $\ell'$ have excess intersection. We then apply Lemma~\ref{lem:excess-intersection} to conclude. Exchanging the roles of $\ell$ and $\ell'$, we may conclude similarly when $\gamma \sim \ell'$.
\end{proof}

\section{Configurations of the first three systoles}\label{sec:propertiesofsecondandthirdsystoles}
In this section we analyse the possible configurations of the first three systoles, continuing in the framework of cross-metric surfaces.
Because of potential multiplicities in the length spectrum, the choice for the first three systoles is not unique. We can nonetheless assert that the shortest non-contractible loop is always simple, whichever we choose. For otherwise it would decompose into a concatenation of shorter non-contractible loops, a contradiction. Even if the second and third systoles cannot always be chosen simple, we show below that we can impose some simple alternatives on these curves.
Basically, given an initial configuration for the first three systoles, we rely on the lemmas in Section~\ref{sec:lemmas} to reconfigure the systoles in order to reach one of the alternatives. We spend most of the proofs justifying that the reconfigurations lead to a valid set of systoles.

\begin{lemma}\label{lem:secondsystole}
  Let $\ell_1$ be a first systole. We may choose the second systole in such a way that it intersects $\ell_1$ at most once and such that it is either simple or homotopic to $\ell_1^2$.
\end{lemma}
\begin{proof}
  Let $\ell_2$ be a second systole that minimizes the number of crossings with $\ell_1$. If $\ell_2$ is not simple then it decomposes into a concatenation of two loops. Each one is non-contractible by the tightness of systoles, so it must be no shorter than $\ell_1$. It follows that $|\ell_2| \geq 2|\ell_1|$ and we may replace  $\ell_2$ by $\ell_1^2$. Note that in this case we may perturb $\ell_2$ to be disjoint from $\ell_1$. We now assume that $\ell_2$ is simple. By way of contradiction suppose that $\ell_1$ and $\ell_2$ have at least two crossings cutting $\ell_1$ into paths $\ell_{1a}, \ell_{1b}$, and cutting $\ell_2$ into paths $\ell_{2a}, \ell_{2b}$. Consider the loop $\gamma:=\ell_{1a}\cup\ell_{2a}$. We can perturb it to satisfy $\card{\gamma\cap \ell_2} < \card{\ell_1\cap \ell_2}$ and $\card{\ell_1\cap \gamma} < \card{\ell_1\cap \ell_2}$. Hence, if $\gamma$ is either contractible or homotopic to $\ell_1$, then by Lemma~\ref{lem:contractible-case} or Lemma~\ref{lem:homotopic-case} there exists a simple tight loop homotopic to $\ell_2$ with strictly fewer crossings with $\ell_1$. This however contradicts the minimal choice of $\ell_2$. The same conclusion holds if we replace $\gamma$ by $\ell_{1i}\cup\ell_{2j}$ for any $i,j\in \{a,b\}$. Choosing for $\ell_{1i}$ the shorter of $\ell_{1a}$ and $\ell_{1b}$ and for 
$\ell_{2j}$ the shorter of $\ell_{2a}$ and $\ell_{2b}$ we obtain a curve that is no longer than $\ell_2$ and neither contractible nor homotopic to $\ell_1$. It may thus serve as a second systole. If $\ell_{1i}\cup\ell_{2j}$ would self-cross then we are back to the first case and may replace it by $\ell_1^2$, contradicting the minimal choice of $\ell_2$. Hence, $\ell_{1i}\cup\ell_{2j}$ is simple and intersects $\ell_1$ fewer times than $\ell_2$ does, again contradicting the choice of $\ell_2$. We conclude that $\ell_1$ and $\ell_2$ cannot have two crossings.
\end{proof}

\begin{lemma}\label{lem:l3-simple}
  We may choose the third systole to be simple or to traverse the first systole twice or thrice.
\end{lemma}
\begin{proof}
  Choose $\ell_3$ with the minimum number of self-crossings among the possible third systoles.
  If $\ell_3$ self-intersects, it decomposes into a concatenation $\alpha\cdot\beta$ of two loops. Since $\ell_3$ is tight, none of $\alpha,\beta$ is contractible. It follows that $\alpha,\beta$ are at least as long as $\ell_1$, so that $|\ell_3|\geq 2|\ell_1|$. Note that $|\ell_2|\leq 2|\ell_1|$ in any case. As a consequence, if $\ell_2$ is not homotopic to $\ell_1^2$ we can choose $\ell_3 = \ell_1^2$ for the third systole. As another consequence, if $|\ell_3|\geq 3|\ell_1|$ we can choose $\ell_3 = \ell_1^3$. We may thus assume that $\ell_2 = \ell_1^2$ and $|\ell_3|<3|\ell_1|$ for the rest of the proof. Now, each of $\alpha$ and $\beta$ is homotopic to $\ell_1$ or to $\ell_2$. For otherwise, it would contradict that $\ell_3$ is the third systole. If $\alpha$ or $\beta$ is homotopic to $\ell_2$, then $|\ell_3| \geq 3|\ell_1|$, contradicting our assumption. So, $\alpha$ and $\beta$ are in fact both homotopic to $\ell_1$.
  \begin{claim}\label{clm:at-most-2}
    $\ell_3$ has at most two self-crossings.
  \end{claim}
  Suppose by way of contradiction that $\ell_3$ has at least three crossings $u,v,w$. Up to a renaming of the crossings their cyclic ordering\footnote{For planar curves, this cyclic ordering is called the \emph{Gauss code} of the curve.} along $\ell_3$ must be one of $uuvvww$, $uuvwwv$, $uuvwvw$, $uvuwvw$, or $uvwuvw$. See Figure~\ref{fig:3-self-crossings}.
  \begin{figure}[h]
    \centering
    \includegraphics[width=\linewidth]{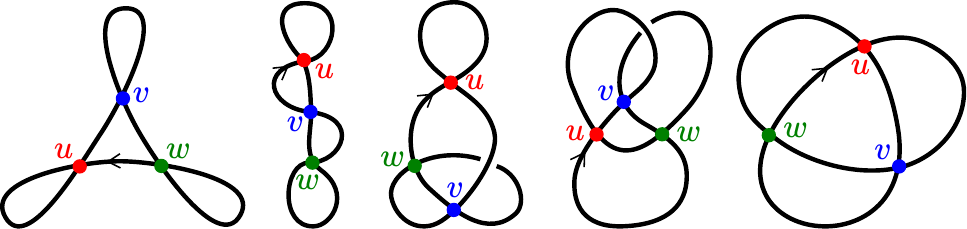}
    \caption{The five possible cyclic orderings of three self-crossings along $\ell_3$.}
    \label{fig:3-self-crossings}
  \end{figure}
  The ordering $uuvvww$ implies that $\ell_3$ decomposes into three loops and some paths connecting them. By tightness of $\ell_3$, none of these loops is contractible, hence no shorter than $\ell_1$. It follows that $\ell_3$ has length at least $3|\ell_1|$, in contradiction with our assumption. Similarly, the ordering $uuvwwv$ implies that $\ell_3$ decomposes into two loops and two paths between $u$ and $w$. If these two paths were homotopic, we could swap them (i.e, perform a smoothing at $u$ and $w$) to obtain a curve homotopic to $\ell_3$, of the same length, with strictly fewer crossings. This would however contradict the fact that $\ell_3$ crosses minimally. It follows that the concatenation of the two paths is a non-contractible curve, hence no shorter than $\ell_1$. In turn, this implies that $\ell_3$ has length at least $3|\ell_1|$, in contradiction with our assumption. The remaining cases can be dealt with similarly. Hence, none of the orderings can occur, proving the claim.
  \begin{claim}\label{clm:at-most-1}
    $\ell_3$ has at most one self-crossing.
  \end{claim}
Suppose on the contrary that $\ell_3$ has at least two crossings. Then, by Claim~\ref{clm:at-most-2} it has exactly two crossings, say $u,v$. These crossings appear in the order $uuvv$ or $uvuv$ along $\ell_3$. See Figure~\ref{fig:2-self-crossings}.
  \begin{figure}[h]
    \centering
    \includegraphics[width=\linewidth]{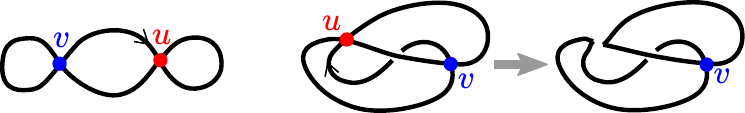}
    \caption{The two possible cyclic orderings of two self-crossings along $\ell_3$.}
    \label{fig:2-self-crossings}
  \end{figure}
  By the same arguments as above, the case $uuvv$ implies that $\ell_3$ has length at least $3|\ell_1|$, which contradicts our assumption. In the other case, $\ell_3$ decomposes as the concatenation of two loops $\alpha\cdot\beta$ with common basepoint $u$. By the discussion before Claim~\ref{clm:at-most-2}, $\alpha$ and $\beta$ are both homotopic to $\ell_1$. Since $u$ is a self-crossing the loops $\alpha$ and $\beta$ touch but do not cross at $u$. It follows that they only cross at $v$, hence have intersection number 1. This however contradicts the fact that they are homotopic, thus proving the claim.

By Claims~\ref{clm:at-most-2} and~\ref{clm:at-most-1}, if $\ell_3$ is not simple it is the concatenation of two simple loops $\alpha$ and $\beta$, both homotopic to $\ell_1$, sharing their basepoint and otherwise disjoint. Thus, $\alpha$ and $\beta$ bound a pinched cylinder. Looking at a small disk neighborhood of the common basepoint of $\alpha$ and $\beta$, we infer that none of them is separating. Equivalently, $\ell_1$ is non-separating. Now, $\alpha$ and $\beta$ cannot be homotopic with fixed basepoint for otherwise $\ell_3$ would be homotopic to $\ell_2=\ell_1^2$. The curve $\alpha\cdot\beta^{-1}$ is thus non-contractible. It is simple (after perturbation) and must be tight, for otherwise it would contradict that $|\ell_2|=2|\ell_1|$. It is null-homologous, hence separating. As a consequence it cannot be homotopic to $\ell_1$ or $\ell_2$. We can thus set $\ell_3= \alpha\cdot\beta^{-1}$.
\end{proof}
  
\begin{lemma}\label{lem:l3-l1-atmostonce}
  We may choose the third systole $\ell_3$ so that it satisfies Lemma~\ref{lem:l3-simple} and in addition it crosses $\ell_1$ at most once.
\end{lemma}
\begin{proof}
  The lemma is trivial if $\ell_3$ is homotopic to $\ell_1^2$ or $\ell_1^3$, since in both cases we can deform $\ell_3$ to make it disjoint from $\ell_1$. By Lemma~\ref{lem:l3-simple} we may thus assume that $\ell_3$ is simple. Among the simple curves that may serve as a third systole (i.e. have the same length as $\ell_3$ and are not homotopic to $\ell_1$ or $\ell_2$), we choose $\ell_3$ so as to minimize the number of crossings with $\ell_1$.  By way of contradiction suppose that $\ell_1$ and $\ell_3$ have at least two crossings $u,v$. Let $u,v$ cut $\ell_1$ into paths $\ell_{1a}, \ell_{1b}$, and cut $\ell_3$ into paths $\ell_{3a}, \ell_{3b}$. Consider the loop $\gamma:=\ell_{1a}\cup\ell_{3a}$. We can perturb it to satisfy $\card{\gamma\cap \ell_3} < \card{\ell_1\cap \ell_3}$ and $\card{\ell_1\cap \gamma} < \card{\ell_1\cap \ell_3}$. Hence, if $\gamma$ is either contractible or homotopic to $\ell_1$ then by Lemma~\ref{lem:contractible-case} or Lemma~\ref{lem:homotopic-case}, there exists a simple tight loop homotopic to $\ell_3$ with strictly fewer crossings with $\ell_1$. This however contradicts the minimal choice of $\ell_3$. Of course, the same conclusion holds if we replace $\gamma$ by $\ell_{1i}\cup\ell_{3j}$ for any $i,j\in \{a,b\}$. We can thus assume that $\ell_{1i}\cup\ell_{3j}$ is neither contractible nor homotopic to $\ell_1$ for any $i,j\in \{a,b\}$.
We now pick $u,v$ in $\ell_1\cap\ell_3$ so that they are consecutive\footnote{By this we mean that $u,v$ appear consecutive along $\ell_3$ and cut $\ell_3$ into paths $\ell_{3a}, \ell_{3b}$ so that the relative interior of $\ell_{3a}$ is disjoint from $\ell_1$.} on $\ell_{3a}$.
If we further assume that $\gamma = \ell_{1a}\cup\ell_{3a}$ is not homotopic to $\ell_2$, then we cannot have $|\ell_{1a}|> |\ell_{3b}|$ since in this case we would have $|\ell_{3b}\cup\ell_{1b}| < |\ell_1|$, implying that $\ell_{3b}\cup\ell_{1b}$ is contractible, which we excluded. It follows that $|\ell_{1a}|\leq |\ell_{3b}|$, whence $|\gamma|\leq |\ell_3|$. We can thus use the simple curve $\gamma$ as a third systole, but this again contradicts the minimal choice of $\ell_3$. It follows that $\gamma$ must be homotopic to $\ell_2$. Of course, the same conclusion holds for any $\gamma'= \ell_{1a}'\cup\ell_{3a}'$, where $\ell_{3a}'$ is a subpath of $\ell_3$ between two consecutive crossings with $\ell_1$, and $\ell_{1a}'$ is any of the two subpaths of $\ell_1$ cut by these crossings. However, the next lemma shows that this is impossible, proving that the existence of two crossings between $\ell_1$ and $\ell_3$ leads to a contradiction.
\end{proof}

\begin{lemma}\label{lem:l3-l1-twice-not-l2}
  Suppose $\ell_3$ is simple and intersects $\ell_1$ at least twice. Suppose in addition that  for any pair of
  crossings cutting $\ell_1$ into paths $\ell_{1a}, \ell_{1b}$, and cutting $\ell_3$ into paths $\ell_{3a}, \ell_{3b}$, the curves $\ell_{1i}\cup\ell_{3j}$, $i,j\in\{a,b\}$,  are neither contractible nor homotopic to $\ell_1$. Then, there must be crossings $u,v$, consecutive on $\ell_{3a}$, such that
$\ell_{1a}\cup\ell_{3a}$ or $\ell_{1b}\cup\ell_{3a}$ is not  homotopic to $\ell_2$. 
 \end{lemma}
\begin{proof}
  Suppose for a contradiction that for any $u,v$ in $\ell_1\cap\ell_3$ that are consecutive on $\ell_{3a}$, the two curves  $\gamma:=\ell_{1a}\cup\ell_{3a}$ and $\gamma':=\ell_{1b}\cup\ell_{3a}$  are homotopic to $\ell_2$.
 Note that the three paths $\ell_{1a}, \ell_{1b}$ and $\ell_{3a}$ have disjoint interiors and only share their endpoints $u$ and $v$. 
We consider a tubular neighborhood $N$ of $\gamma$, sufficiently small so that $\ell_{1b}$ intersects $N$ along two paths $uu'$ and $vv'$. We then replace the subpath $u'u\cdot \ell_{3a}\cdot v'v$ of $\gamma'$ by a homotopic path $u'v'$ in $N$ that crosses $\gamma$ minimally, hence at most once. We still denote by $\gamma'$ the resulting cycle. Since $\gamma$ and $\gamma'$ are homotopic, their algebraic intersection number is null, and since they cross at most once, they must be disjoint. It follows that the simple loop $u'u\cdot \ell_{3a}\cdot vv'\cdot v'u'$ bounds a disk $B$ (in $N$) and that the homotopic and disjoint simple curves $\gamma$ and $\gamma'$ bound a cylinder $A$ in $S$. There are two possibilities. Refer to Figure~\ref{fig:homotopic_l2_case}.
\begin{figure}[ht]
  \centering
  \includesvg[\linewidth]{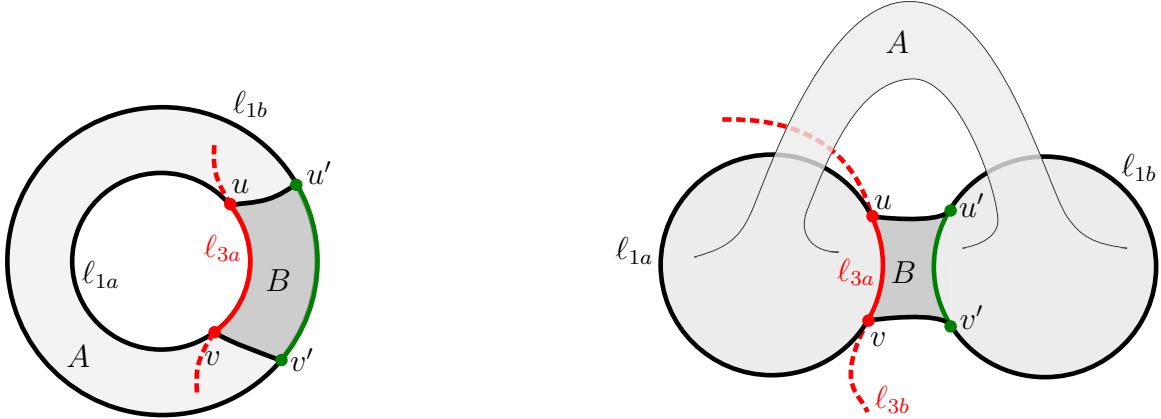}
  \caption{Illustration for the cases where $B$ lies respectively inside (left) or outside (right) $A$.}
  \label{fig:homotopic_l2_case}
\end{figure}

\begin{itemize}
 \item Either $B$ lies inside $A$. In this case $A\setminus B$ is a disk, implying that $\ell_{1a}\cup \ell_{1b}$ bounds a disk. However, this contradicts the fact that $\ell_1$ is non-trivial.
 \item Or $B$ lies outside $A$. In this case, since the crossings of $\ell_3$ and $\ell_1$ are transverse, it must be that $\ell_{3b}$ is locally outside $A\cup B$ at $u$ and $v$.
Let $p$ be the subpath of $\ell_{3b}$ between $v$ and $w$, the first crossing with $\ell_1$ when traversing $\ell_{3b}$ from $v$ to $u$. (Possibly $w=u$ if $\ell_3$ and $\ell_1$ have only two crossings.) See Figure~\ref{fig:parallel-pi}. The endpoints of $p$ divide $\ell_1$ into two subpaths, so that each of them defines a loop with $p$ which is homotopic to $\ell_2$ by our hypothesis. Denote by $\gamma''$ one of these two loops. Note that $\gamma''$  does not enter the interior of $A\cup B$, a torus with one boundary. In particular $\intForm{\gamma''}{\pi}$ is null, where $\pi$ is a parallel of this torus as pictured in Figure~\ref{fig:parallel-pi}.
     \begin{figure}[ht]
       \centering
       \includesvg[.6\linewidth]{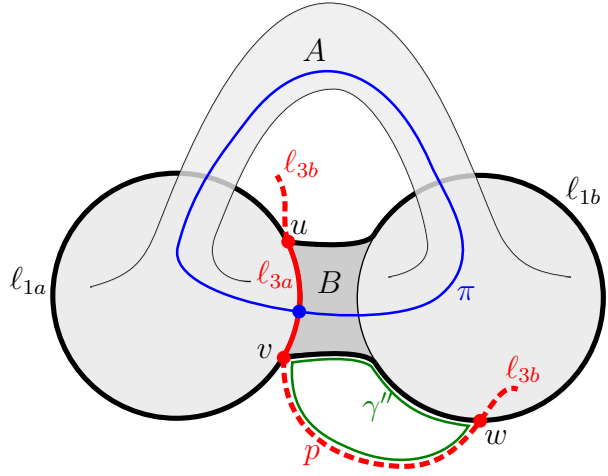}
       \caption{$\gamma''$ has algebraic intersection zero with the parallel $\pi$.}
       \label{fig:parallel-pi}
     \end{figure}
     However, this is in contradiction with the fact that $\intForm{\gamma}{\pi} = \pm 1$, while  $\intForm{\gamma''}{\pi} = \intForm{\gamma}{\pi}$, since $\gamma$ and $\gamma''$ are both homotopic to $\ell_2$ by our initial hypothesis.
\end{itemize}
 In conclusion, we have proved that it is not possible that every concatenation of a piece of $\ell_1$ with a piece of $\ell_3$ between consecutive crossings is homotopic to $\ell_2$.
\end{proof}
In the following, if $u$ and $v$ are two distinct crossings of $\ell_2$ and $\ell_3$ cutting them into $\ell_2 = \ell_{2a}\cup\ell_{2b}$ and $\ell_3 = \ell_{3a}\cup\ell_{3b}$, we put $\gamma_{ij}=\ell_{2i}\cup\ell_{3j}$ for $i,j\in \{a,b\}$.
\begin{lemma}\label{lem:l3-l2-twice}
  Suppose that $\ell_2$ is simple and crosses $\ell_1$ at most once. Let $\ell_3$ be a third systole that crosses $\ell_2$ minimally among the third systoles that cross $\ell_1$ at most once. Then, if $\ell_3$ is simple it crosses $\ell_2$ at most twice.
\end{lemma}
 \begin{proof}
Suppose that the third systole $\ell_3$ chosen as in the lemma is simple and crosses $\ell_2$ at least twice. Let $u,v$  be two crossings of $\ell_2$ and $\ell_3$ cutting $\ell_2$ into paths $\ell_{2a}, \ell_{2b}$, and cutting $\ell_3$ into paths $\ell_{3a}, \ell_{3b}$. The loop $\gamma_{aa}:=\ell_{2a}\cup\ell_{3a}$ can be perturbed to satisfy $\card{\gamma_{aa}\cap \ell_3} < \card{\ell_2\cap \ell_3}$. Hence, if $\gamma_{aa}$ is either contractible or homotopic to $\ell_2$, then by Lemma~\ref{lem:contractible-case} and Lemma~\ref{lem:excess-intersection} (respectively), there exists a simple tight loop homotopic to $\ell_3$ with strictly fewer crossings with $\ell_2$ and at most one crossing with $\ell_1$. This however contradicts the minimal choice of $\ell_3$. The same conclusion holds if we replace $\gamma_{aa}$ by $\gamma_{ij}:=\ell_{2i}\cup\ell_{3j}$ for any $i,j\in \{a,b\}$. We thus have
   \begin{claim}\label{clm:not-1-not-l2}
     $\gamma_{ij}$ is neither contractible nor homotopic to $\ell_2$ for any $i,j\in \{a,b\}$. 
   \end{claim}
   Since $\ell_2$ and $\ell_3$ may have more than two crossings the curve $\gamma_{ij}$ may be self-crossing. We can however assert that it is simple whenever it is homotopic to $\ell_1$.
   \begin{claim}\label{clm:l1-simple}
     If $\gamma_{ij}$ is homotopic to $\ell_1$, then it must be simple. Moreover, $\gamma_{ij}$ and $\ell_1$ are either disjoint or have at least two crossings.
   \end{claim}
 To prove the claim we may proceed similarly to the proof of Lemma~\ref{lem:excess-intersection}.
   Since $\gamma_{ij}$ is homotopic to the simple curve $\ell_1$, it must have an embedded bigon or monogon whenever it self-intersects. Such a bigon has one side supported by $\ell_2$ and the other side supported by $\ell_3$ since $\ell_{2i}$ and $\ell_{3j}$ are simple paths. Similarly, the only way $\gamma_{ij}$ can have a monogon is when this monogon is actually a bigon of $\ell_2$ and $\ell_3$. In both cases, we can switch the $\ell_3$ side along the $\ell_2$ side of the bigon to decrease the number of crossings of $\ell_2$ and $\ell_3$. This leaves $\ell_3$ simple without changing its number of crossings with $\ell_1$. This however contradicts the minimal choice of $\ell_3$. It follows that $\gamma_{ij}$ is simple. Moreover, having $\gamma_{ij}$ homotopic to $\ell_1$ implies that their algebraic intersection number is zero, preventing them from having exactly one crossing. This ends the proof of the claim.

   We now suppose for a contradiction that $\ell_2$ and $\ell_3$ have at least three crossings. Since $\ell_3$ cuts $\ell_1$ at most once, we can find two crossings $u$ and $v$ of $\ell_2$ and $\ell_3$ so that
   \begin{itemize}
   \item they appear consecutive on a subpath $\ell_{3a}$ of $\ell_3$,
     \item $\ell_{3a}$ does not cross $\ell_1$, and
     \item $|\ell_{3a}|\leq |\ell_{3b}|$, where as usual $\ell_3=\ell_{3a}\cup \ell_{3b}$.
     \end{itemize}
     The crossings $u$ and $v$ also cut $\ell_2$ into $\ell_2=\ell_{2a}\cup \ell_{2b}$, and we assume without loss of generality that  $|\ell_{2a}| \leq |\ell_{2b}|$.
     It follows that $\gamma_{aa}=\ell_{2a}\cup \ell_{3a}$ is a simple curve crossing $\ell_1$ at most once and satisfying $|\gamma_{aa}|\leq |\ell_3|$ and $\card{\gamma_{aa}\cap \ell_2} < \card{\ell_3\cap \ell_2}$ after perturbation. By Claim~\ref{clm:not-1-not-l2}, the curve $\gamma_{aa}$ is neither contractible nor homotopic to $\ell_2$. It is thus homotopic to $\ell_1$, for otherwise it could serve as a third systole with fewer crossings with $\ell_2$ than $\ell_3$, which is in contradiction with the choice of $\ell_3$. We infer by Claim~\ref{clm:l1-simple} that $\gamma_{aa}$ is disjoint from $\ell_1$. There are two possibilities according to whether $\ell_{2b}$ is shorter than $\ell_{3b}$ or not.
     \begin{itemize}
     \item If $|\ell_{2b}|\leq |\ell_{3b}|$, then $|\gamma_{ba}|\leq |\ell_{3}|$. An argument analogous to the above one for $\gamma_{aa}$ implies that $\gamma_{ba}$ and $\ell_1$ are homotopic and disjoint. We infer that $\ell_2$ is also disjoint from $\ell_1$. Arguing as in the proof of Lemma~\ref{lem:l3-l1-twice-not-l2} we deduce that $\ell_2$ bounds a genus one subsurface $S_1$ containing $\ell_1$ and $\ell_{3a}$. See right Figure~\ref{fig:homotopic_l2_case}, where the roles of $\ell_1$  and  $\ell_2$ should be exchanged.

       We then consider the shorter of the two components of $\ell_3\setminus \ell_2$ following or preceding $\ell_{3a}$ along $\ell_3$. Denote the closure of this component by $\ell'_{3a}$ and denote by $\ell'_{3b}$ the closure of $\ell_3\setminus \ell'_{3a}$. We thus have $|\ell'_{3a}|\leq |\ell'_{3b}|$. Note that the relative interior of $\ell'_{3a}$ lies outside $S_1$. The common endpoints of $\ell'_{3a}$ and $\ell'_{3b}$ also cut $\ell_2$ into $\ell_2=\ell'_{2a}\cup \ell'_{2b}$ so that  $|\ell'_{2a}| \leq |\ell'_{2b}|$. The curve $\gamma'_{aa} = \ell'_{2a}\cup\ell'_{3a}$ is thus no longer than $\ell_3$, cuts $\ell_1$ at most once and, after perturbation, has strictly fewer crossings with $\ell_2$ than $\ell_3$. By Claim~\ref{clm:not-1-not-l2} and the minimal choice of $\ell_3$ we infer that $\gamma'_{aa}$ is homotopic to $\ell_1$. Considering a parallel $\pi$ of $S_1$ as in Figure~\ref{fig:parallel-pi} (where $\ell_1$ should be replaced by $\ell_2$), we now have that $\pi$ is disjoint from $\gamma'_{aa}$ but cuts $\gamma_{aa}$ once. This however leads to a contradiction since the homotopic curves $\gamma_{aa}$ and $\gamma'_{aa}$ should define the same intersection form.
       
     \item If $|\ell_{2b}|>|\ell_{3b}|$, then $|\gamma_{ab}|<|\ell_2|$. By Claim~\ref{clm:not-1-not-l2}, this directly implies that $\gamma_{ab}$ is homotopic to $\ell_1$. In turn, Claim~\ref{clm:l1-simple} implies that $\gamma_{ab}$ is simple and disjoint from $\ell_1$. Together with the hypothesis that $\ell_{3a}$ is disjoint from $\ell_1$, we infer that $\ell_3$ and $\ell_1$ are disjoint. Similarly to the above case, we deduce that $\ell_3$ bounds a genus one subsurface $S_2$ containing $\ell_1$ and $\ell_{2a}$. We can now repeat verbatim the same argument as above, exchanging the roles of $\ell_2$ and $\ell_3$, to conclude that a curve $\gamma'_{aa}$ is homotopic to $\ell_1$ but disjoint from the interior of $S_2$. Then, the algebraic intersection number of a parallel of $S_2$ with $\gamma'_{aa}$ and with $\gamma_{aa}$ would be different, leading to a contradiction.
     \end{itemize}
     Since none of the above possibilities can occur, we conclude that $\ell_3$ and $\ell_2$ have at most two crossings.
\end{proof}
\begin{lemma}\label{lem:l3-l2-exactlytwice}
  Suppose that $\ell_1, \ell_2, \ell_3$ are as in Lemma~\ref{lem:l3-l2-twice} and suppose furthermore that $\ell_3$ is simple and crosses $\ell_2$ exactly twice. Then, $\ell_1$ is disjoint from $\ell_2$ and $\ell_3$, and  is contained in a genus one subsurface bounded by either $\ell_2$ or $\ell_3$.  Moreover, $\ell_2$ and $\ell_3$ being cut by their two crossing points into $\ell_2 = \ell_{2a}\cup\ell_{2b}$ and $\ell_3 = \ell_{3a}\cup\ell_{3b}$, if $\ell_2$ bounds then $\gamma_{aa}$ and $\gamma_{ba}$ are both homotopic to $\ell_1$  (assuming $\ell_{3a}$ lies in the genus one subsurface), while if $\ell_3$ bounds then $\gamma_{aa}$ and $\gamma_{ab}$ are both homotopic to $\ell_1$  (assuming $\ell_{2a}$ lies in the genus one subsurface).
\end{lemma}
\begin{proof}
Since $\ell_1$ cuts each of $\ell_2$ and $\ell_3$ at most once, we may assume that $\ell_{2a}\cap \ell_1=\ell_{3a}\cap \ell_1=\emptyset$. As already noticed, the curve $\ell_2$ has strictly more crossings with $\ell_3$ than with any of the curves $\gamma_{ij}=\ell_{2i}\cup\ell_{3j}$ for $i,j\in \{a,b\}$. Also note that Claim~\ref{clm:not-1-not-l2} and Claim~\ref{clm:l1-simple} in the proof of Lemma~\ref{lem:l3-l2-twice} apply, since $\ell_2$ and $\ell_3$ have two crossings. For convenience, we  set $\bar{a}=b$ and $\bar{b}=a$. 
  \begin{claim}\label{clm:not-3-among-4}
    For any $i,j\in\{a,b\}$, we cannot have simultaneously $|\gamma_{i\bar{j}}|, |\gamma_{\bar{i}j}|, |\gamma_{\bar{i}\bar{j}}| \leq |\ell_3|$.
  \end{claim}
  If $i=a$, then $\card{\gamma_{i\bar{j}}\cap \ell_1} \leq 1$. It follows by Claim~\ref{clm:not-1-not-l2} and the minimal choice of $\ell_3$ that $\gamma_{i\bar{j}}$ is homotopic to $\ell_1$. Then, Claim~\ref{clm:l1-simple} implies $\ell_{3\bar{j}}\cap\ell_1=\emptyset$. We infer by the same kind of arguments that $\gamma_{\bar{i}\bar{j}}$ is homotopic to $\ell_1$. In turn, we deduce that $\ell_{2\bar{i}}\cap\ell_1=\emptyset$, whence $\gamma_{\bar{i}j}$ is also homotopic to $\ell_1$. From $\gamma_{i\bar{j}}$ and $\gamma_{\bar{i}\bar{j}}$ being both homotopic to $\ell_1$, we infer that $\ell_2$ bounds a genus one subsurface $S'$ containing $\ell_1$ and $\ell_{3\bar{j}}$ and whose interior is disjoint from $\gamma_{\bar{i}j}$. This however leads to a contradiction by considering the algebraic intersection number of these three curves with a parallel of $S'$. The case when $i=b$ can be dealt with similarly, also leading to a contradiction. This ends the proof of Claim~\ref{clm:not-3-among-4}.

  Observing that $|\gamma_{ij}|+|\gamma_{\bar{i}\bar{j}}|= |\ell_2|+|\ell_3|$, we deduce from Claim~\ref{clm:not-3-among-4} that either
  \begin{enumerate}
  \item\label{it:1} $|\gamma_{aa}|, |\gamma_{ab}| > |\ell_3|$ and $|\gamma_{bb}|, |\gamma_{ba}| < |\ell_2|$, or
  \item\label{it:2} $|\gamma_{aa}|, |\gamma_{ba}| > |\ell_3|$ and $|\gamma_{bb}|, |\gamma_{ab}| < |\ell_2|$, or
  \item\label{it:3} $|\gamma_{bb}|, |\gamma_{ab}| > |\ell_3|$ and $|\gamma_{aa}|, |\gamma_{ba}| < |\ell_2|$, or
  \item\label{it:4} $|\gamma_{bb}|, |\gamma_{ba}| > |\ell_3|$ and $|\gamma_{aa}|, |\gamma_{ab}| < |\ell_2|$.
  \end{enumerate}
  Claim~\ref{clm:not-1-not-l2} and Claim~\ref{clm:l1-simple} imply that each curve $\gamma_{ij}$ shorter than $\ell_2$ must be homotopic to $\ell_1$ and disjoint from $\ell_1$. In Cases \ref{it:1} and \ref{it:4} we directly infer that $\ell_3$ is disjoint from $\ell_1$. Exchanging the roles of $\ell_2$ and $\ell_1$, we deduce from the proof of Lemma~\ref{lem:l3-l1-twice-not-l2} that $\ell_3$ bounds a genus one subsurface $S'$ containing $\ell_1$. In Case~\ref{it:1}  we also deduce that  $S'$ contains $\ell_{2b}$ and that $\ell_1$ is disjoint from $\ell_{2b}$. It follows that $\ell_1$ is also disjoint from $\ell_2$. In Case~\ref{it:4} we get that $\ell_{2a}$ lies inside $S'$. Since $\ell_{2b}$ does not intersect the interior of $S'$, we also deduce that $\ell_1$ is disjoint from $\ell_2$. The situation is depicted in Figure~\ref{fig:2crossings}.
  \begin{figure}[ht]
    \centering
    \includegraphics[width=.85\linewidth]{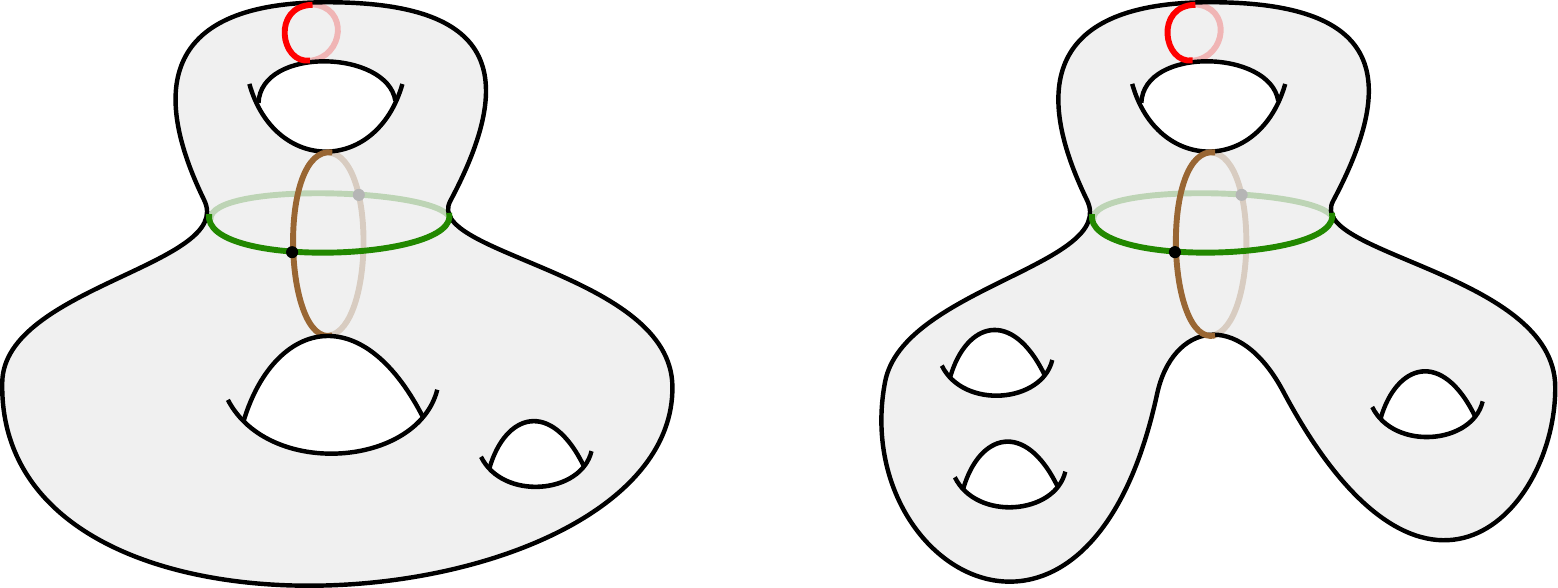}
    \caption{$\ell_2$ and $\ell_3$ have exactly two crossings and are both disjoint from the red curve $\ell_1$. Depending on their respective lengths, each of the crossing curves may play the role of $\ell_2$ or $\ell_3$. There are two configurations according to whether the union of the red and brown curves disconnect the surface or not.}
    \label{fig:2crossings}
  \end{figure}
  Cases \ref{it:2} and \ref{it:3} can be dealt with similarly to conclude that $\ell_1$ is disjoint from both $\ell_2$ and $\ell_3$. 
  We now see \emph{a posteriori} that $\ell_{2a}$ and $\ell_{2b}$ can be interchanged and similarly for $\ell_{3a}$ and $\ell_{3b}$. It is thus sufficient to consider Cases~\ref{it:3} and~\ref{it:4} in the Lemma.
\end{proof}

\section{Computing the second and third systoles}\label{sec:computingsecondthirdsystole}

Before we turn to the effective computation of the second and third systoles, let us state a simple lemma that will prove useful in the proof of Theorem~\ref{thm:secondsystole}.

\begin{lemma}\label{lem:not-homotopic-to-boundaries}
  Let $S$ be a surface with boundary and let $B$ be a subset consisting of $k$ boundary components. A simple curve $\gamma$ on $S$ is not homotopic to any curve in $B$ or to a point if and only if
  \begin{enumerate}
  \item $\gamma$ is non-contractible in the surface $\bar{S}$ obtained by capping off the boundary curves in $B$ with disks, or
  \item\label{it:boundaries2} $\gamma$ is essential in $S$.
  \end{enumerate}
   When $k=1$, Condition~\ref{it:boundaries2} is redundant.
\end{lemma}
\begin{proof}
  If $\gamma$ is contractible or homotopic to a curve in $B$, then it bounds a disk or an annulus with that curve. It follows that $\gamma$ bounds a disk in $\bar{S}$, hence is contractible. Of course, $\gamma$ is inessential in $S$ in this case. Conversely, if $\gamma$ is contractible in $\bar{S}$, and is  inessential in ${S}$, then it is either contractible in $S$, or homotopic to a boundary component of $S$, which must lie in $B$ since it is contractible in $\bar{S}$. If $k=1$, then an essential curve in $S$ must be non-contractible in $\bar{S}$, so that the second condition is superfluous.
\end{proof}

We now start with an algorithm for computing a second systole, repeating for convenience the statement of Theorem~\ref{thm:secondsystole}. For algorithmic purposes we use the combinatorial framework. By a simple curve we actually mean a weakly simple curve. Similarly we may speak of disjoint curves when they are actually weakly disjoint. 
\bigskip

\noindent      
{\sc T{\footnotesize HEOREM} 1.1.} 
{\it Let $G$ be a weighted graph of complexity $n$ cellularly embedded on a surface $S$ of genus $g$ with $b$ boundary components. A second systole of $G$ can be computed in $O(n^2\log n)$ time or in $O((b^2+g^3)n\log n)$ time.
} 

\begin{proof}
    Denote the combinatorial surface $(S,G)$ by $\Sigma$. We start by computing a first systole $\ell_1$ of $\Sigma$. Using one of the known algorithms~\cite[Th. 4.7]{ew-csec-10}, this can be done in time $O(n^2\log n)$ or $O((b+g^3)n\log n)$. By Lemma~\ref{lem:secondsystole}, we may choose the second systole in such a way that it intersects $\ell_1$ at most once and such that it is either simple or homotopic to $\ell_1^2$. Hence, we only need to describe an algorithm for the former case of a shortest simple cycle that is not homotopic to $\ell_1$.

There are three cases to consider according to whether  $\ell_1$ and $\ell_2$ cross or not and to whether $\ell_1$ is separating or not.
\begin{enumerate}
\item \textbf{If $\ell_1$ is separating}, then $\ell_2$ cannot cross $\ell_1$.
Let $\Sigma_1$ and $\Sigma_2$ be the two components of $\Sigma\cut \ell_1$, and denote by $\beta_i$ the boundary component of $\Sigma_i$ corresponding to $\ell_1$, for $i=1,2$. Also denote by $\Gamma_i$, for $i=1,2$,  the set of (closed) curves in $\Sigma_i$ that are simple, non-contractible and not homotopic to $\beta_i$. The set of curves in $\Sigma$ that are simple, non-contractible, disjoint from $\ell_1$ and not homotopic to $\ell_1$ identifies with $\Gamma_1\cup\Gamma_2$. Indeed, since $\ell_1$ is not contractible, two curves in $\Sigma_i$ are homotopic in $\Sigma$ if and only if they are homotopic in $\Sigma_i$. By Lemma~\ref{lem:not-homotopic-to-boundaries}, $\Gamma_i$ identifies with the simple and non-contractible curves in the surface $\bar{\Sigma_i}$ obtained by capping off $\beta_i$ with a disk. It follows that we can take for $\ell_2$ the shorter of $\gamma_1$ and $\gamma_2$, where $\gamma_i$ is a shortest non-contractible curve in $\bar{\Sigma_i}$. A second systole can thus be computed in time $O(n^2\log n)$~\cite[Cor. 5.3]{erickson2004} or $O((b+g^3)n\log n)$~\cite[Th. 5.5]{ew-csec-10}.

\item \textbf{If $\ell_1$ is non-separating and does not cross $\ell_2$}, we first claim that two (closed) curves $\gamma_1,\gamma_2$ not homotopic to any copy of  $\ell_1$ in $\Sigma\cut\ell_1$ are homotopic in $\Sigma$ if and only if they are homotopic in $\Sigma\cut\ell_1$. Indeed, any homotopy between  $\gamma_1$ and $\gamma_2$ in $\Sigma\cut\ell_1$ quotient to a homotopy in $\Sigma$. In the other direction, if $\gamma_1\sim \gamma_2\not\sim \ell_1$ in $\Sigma$ then one can make $\gamma_1$ and $\gamma_2$ disjoint by swapping bigons according to \cite[Lem. 3.1]{hs-ics-85}. Each bigon swapping can actually be performed in   $\Sigma\cut\ell_1$, since it is not crossed by $\ell_1$. Now,  being disjoint and homotopic, $\gamma_1$ and $\gamma_2$ bound a cylinder in $\Sigma$. This cylinder cannot intersect $\ell_1$ since $\gamma_1$ and $\gamma_2$ are disjoint from $\ell_1$ and not homotopic to $\ell_1$. It follows that this cylinder lies in $\Sigma\cut\ell_1$, showing that $\gamma_1$ and $\gamma_2$ are homotopic in $\Sigma\cut\ell_1$. This proves the claim. We let $\Sigma'$ be the surface obtained from $\Sigma\cut\ell_1$ by capping of with a disk the boundary components corresponding to $\ell_1$. It follows from the claim and  Lemma~\ref{lem:not-homotopic-to-boundaries}, that a simple curve in $\Sigma$ is non-contractible, disjoint from $\ell_1$  and not homotopic to $\ell_1$ if and only if it is either non-contractible in $\Sigma'$, or essential in $\Sigma\cut\ell_1$.\footnote{The sets of non-contractible simple curves in $\Sigma'$ and of essential curves in $\Sigma\cut\ell_1$ decompose into the disjoint union of the essential curves in $\Sigma\cut\ell_1$ and the simple curves in $\Sigma$ homotopic to any boundary curve. In addition to the computation of shortest essential cycles one could thus use an algorithm for computing a shortest cycle homotopic to a given curve as in~\cite{de2010} in order to compute $\ell_2$. However, this would not decrease the asymptotic complexity of our algorithm.}

  Let $c$ be a shortest  essential (weakly)  simple curve in $\Sigma\cut\ell_1$ and let $c'$ be a shortest non-contractible curve in $\Sigma'$. From the previous discussion we can take for $\ell_2$ a shortest curve among $c$ and $c'$. It can be computed in $O(n^2\log n)$ time~\cite[Th. 4.7]{ew-csec-10}, or in $O((b^2+g^3)n\log n)$ time according to Theorem 5.5 in the same paper.
  
\item \textbf{If $\ell_1$ does cross $\ell_2$ once}, then $\ell_1$ must be non-separating. The surface $\Sigma\cut\ell_1$  is thus connected  and must have two of its boundary components correspond to copies of $\ell_1$. For each vertex on one of these boundary components we compute a shortest path in $\Sigma\cut\ell_1$ to the corresponding vertex on the other boundary component. A shortest among these paths, say $\ell$, corresponds to a shortest (weakly) simple cycle in $\Sigma$ that intersects $\ell_1$ exactly once. Since the algebraic intersection number of $\ell$ and $\ell_1$ is one, it follows that $\ell$ is not homotopic to $\ell_1$. Computing a shortest path with a given source and target can be done in time $O(n\log n)$ with Dijkstra's algorithm and $\ell_1$ has $O(n)$ vertices, so this step takes time $O(n^2\log n)$. This complexity can be decreased to $O(gn\log n)$ by noting that a shortest path between the two copies of $\ell_1$ in $\Sigma\cut\ell_1$ is a shortest path between the same copies after capping off with a disk all its boundary components. This allows us to apply the shortest path algorithm of Cabello et al.~\cite[Th. 4.4 and Cor. 4.5]{cce-msspe-13}. Indeed, after $O(gn\log n)$ time precomputation, it provides for every vertex on the face boundary corresponding to a copy of $\ell_1$ the shortest path distance to any other vertex in $O(\log n)$ time. We can thus compute the shortest path distance between all the pairs of corresponding vertices on the copies of $\ell_1$ in $O(n\log n)$ time. We finally compute the corresponding shortest path in $O(n\log n)$ using Dijkstra's algorithm. In total, this approach provides $\ell_2$ in $O(gn\log n)$ time. 
\end{enumerate}
	
	Finally, a second systole is obtained by taking a  shortest curve among the curves computed in the three cases above, whichever applies. 
\end{proof}

Next we present an algorithm for computing the third systole. Again, we repeat the statement of Theorem~\ref{thm:thirdsystole} for convenience.
\bigskip

\noindent      
{\sc T{\footnotesize HEOREM} 1.2.} 
{\it	Let $G$ be a weighted graph of complexity $n$ cellularly embedded on a surface $S$ of genus $g$. The third systole of $G$ can be computed in $O(n^3)$ time.
  } 

\begin{proof}
    Denote the combinatorial surface $(S,G)$ by $\Sigma$. We start by computing a first systole $\ell_1$~\cite[Th. 4.7]{ew-csec-10} and a second systole $\ell_2$ (Theorem~\ref{thm:secondsystole}) of $\Sigma$ in time $O(n^2\log n)$. By Lemma~\ref{lem:l3-simple}, we may choose the third systole $\ell_3$ such that it is either simple or traverses $\ell_1$ twice or thrice. It follows that we only need to describe an algorithm for the former case of a shortest \emph{simple} cycle that is neither homotopic to $\ell_1$ nor to $\ell_2$. By Lemmas~\ref{lem:secondsystole},~\ref{lem:l3-l1-atmostonce},~\ref{lem:l3-l2-twice} and~\ref{lem:l3-l2-exactlytwice}, we may assume that the third systole $\ell_3$ satisfies the conditions of one of the cases described below. Hence, it suffices to compute for each of these cases the shortest cycle that satisfies the conditions and that is not homotopic to $\ell_1$ or $\ell_2$. A third systole is then obtained by taking the shortest of the cycles obtained in the different cases. 
    
    \begin{enumerate}
        \item \textbf{Assume that $\ell_3$ does not cross $\ell_1$ or $\ell_2$}. Let $\Sigma'$ be the surface obtained from $\Sigma\cut(\ell_1\cup\ell_2)$ by capping off with a disk the boundary components corresponding to $\ell_1\cup \ell_2$. By a similar reasoning as in Cases 1 and 2 in the proof of Theorem~\ref{thm:secondsystole}, it suffices to compute the shortest essential curve in $\Sigma\cut(\ell_1\cup\ell_2)$ and the shortest non-contractible curve in $\Sigma'$. These can be computed in $O(n^2\log n)$ time~\cite[Th. 4.7]{ew-csec-10}.
        
        \item \textbf{Assume that $\ell_1$ and $\ell_2$ do not cross and that $\ell_3$ crosses at least one of them exactly once}. By simple arguments based on the algebraic intersection number it is seen that either (a) $\ell_3$ intersects both $\ell_1$ and $\ell_2$ exactly once, or (b) intersects only one of the two, say $\ell_i$, in which case the uncrossed cycle $\ell_{j\neq i}$ does not separate the two copies of $\ell_i$ in $\Sigma\cut\ell_i$. Case (a) can only occur if both $\ell_1$ and $\ell_2$ are non-separating. In this case, we compute a shortest cycle $c_1$ that intersects $\ell_1$ exactly once, ignoring $\ell_2$. Note that $c_1$ has algebraic intersection number $\pm 1$ with $\ell_1$. Hence, it cannot be homotopic to $\ell_1$ or to $\ell_2$ since these last two have algebraic intersection number 0. We also compute a shortest cycle $c_2$ that intersects $\ell_2$ exactly once, ignoring $\ell_1$. We then return the shorter of $c_1$ and $c_2$. In the second case (b) we return a shortest cycle that intersects $\ell_i$ once in the surface $\Sigma\cut\ell_j$. Note that we only need to consider Case (b) when $\ell_i$ is non-separating and $\ell_j$ is separating but does not separate the two copies of $\ell_i$ in $\Sigma\cut\ell_i$. Indeed, the computations in Case (a) actually include Case (b) when both $\ell_1$ and $\ell_2$ are non-separating. Checking the conditions for the two cases can be done in linear time. The rest of the computation, already discussed in Case 3 in the proof of Theorem~\ref{thm:secondsystole} takes $O(n^2\log n)$ time by repeated use of Dijkstra's algorithm, or $O(gn\log n)$ time with the approach of Cabello et al.~\cite{cce-msspe-13}.
        
        \item \textbf{Assume that $\ell_1$ and $\ell_2$ cross exactly once and that $\ell_3$ crosses exactly one of $\ell_1$ and $\ell_2$ exactly once and does not cross the other}. In this case, we consider the surface $\Sigma\cut(\ell_1\cup\ell_2)$ obtained after cutting along both $\ell_1$ and $\ell_2$. This surface has a boundary component $\delta$ composed alternately of two copies of $\ell_1$ and two copies of $\ell_2$. We compute for each vertex $v$ on $\delta$ a shortest path to its corresponding copy on $\delta$. (Some care is needed when $v$ is one of the four vertices corresponding to the crossing
          of $\ell_1$ and $\ell_2$. In this case, we must ignore the ``opposite'' copy of $v$ in $\delta$ to ensure that the connecting paths correspond to cycles that cross only one of $\ell_1$ and $\ell_2$.) The shortest of these paths corresponds to the shortest cycle in $\Sigma$ that crosses exactly one of $\ell_1$ and $\ell_2$ exactly once and does not cross the other. As before, the algebraic intersection numbers allow to check that this cycle is not homotopic to $\ell_1$, $\ell_2$ or a point. Analogously to Case 2, the whole computation can be performed  in $O(n^2\log n)$ time by repeated use of Dijkstra's algorithm, or $O(gn\log n)$ time with the approach of Cabello et al.~\cite{cce-msspe-13}.
        
        \item \textbf{Assume that $\ell_1$ and $\ell_2$ cross exactly once and that $\ell_3$ crosses both $\ell_1$ and $\ell_2$ exactly once}. In this case, we consider again the surface $\Sigma\cut(\ell_1\cup\ell_2)$, which has a boundary component $\delta$ composed of two copies of $\ell_1$ and two copies of $\ell_2$. For each vertex on either copy of $\ell_1$ along $\delta$ we compute shortest paths in $\Sigma\cut(\ell_1\cup\ell_2)$ to all vertices on either copy of $\ell_2$. 
          We divide these shortest paths into pairs: the shortest path from a vertex $v_1$ on one of the copies of $\ell_1$ to a vertex $v_2$ on one of the copies of $\ell_2$ is paired with the shortest path from $v'_1$ to $v'_2$, where $v'_1$ is the vertex corresponding to $v_1$ on the other copy of $\ell_1$ and $v'_2$ is the vertex corresponding to $v_2$ on the other copy of $\ell_2$. The shortest of these pairs of paths corresponds to the shortest cycle in $\Sigma$ that intersects $\ell_1$ and $\ell_2$ each exactly once.
(Some care is needed when $v_1$ or $v_2$ is one of the four vertices in $\delta$ corresponding to the crossing $u$
          of $\ell_1$ and $\ell_2$. For instance, if $v_1$ and $v_2$ are ``opposite'' copies of $u$ in $\delta$, then the shortest path connecting them actually corresponds to a cycle crossing both $\ell_1$ and $\ell_2$, and there is no need to  pair it with another shortest path.)

          As before, it is straightforward to check that this cycle is not homotopic to $\ell_1$ and $\ell_2$ by checking the algebraic intersection numbers. Computing shortest paths from a given vertex on either copy of $\ell_1$ to all vertices on either copy of $\ell_2$ can be done in time $O(n\log n)$ using Dijkstra's algorithm. Since the two copies of $\ell_1$ have $O(n)$ vertices, this step takes time $O(n^2\log n)$ in total.
        
        \item \textbf{Assume that $\ell_2$ and $\ell_3$ cross exactly twice and that both of them do not cross $\ell_1$}. Then we are in the situation of Lemma~\ref{lem:l3-l2-exactlytwice} depicted in Figure~\ref{fig:2crossings}, where $\ell_2$ can either be the (separating) green cycle or the (non-separating) brown cycle. 
        
          Assume first that $\ell_2$ is the green cycle. In this case, we choose a vertex $v$ on $\ell_1$ and compute a shortest path $p$ between its two copies $v', v''$ in the component $\cal P$ of $\Sigma\cut (\ell_1\cup \ell_2)$ bounded by the two copies of $\ell_1$ and one copy of $\ell_2$. This component is a disk with two smaller disks removed, i.e., a pair of pants. See Figure~\ref{fig:pair-of-pants}.
          \begin{figure}[h]
            \centering
            \includesvg[.8\linewidth]{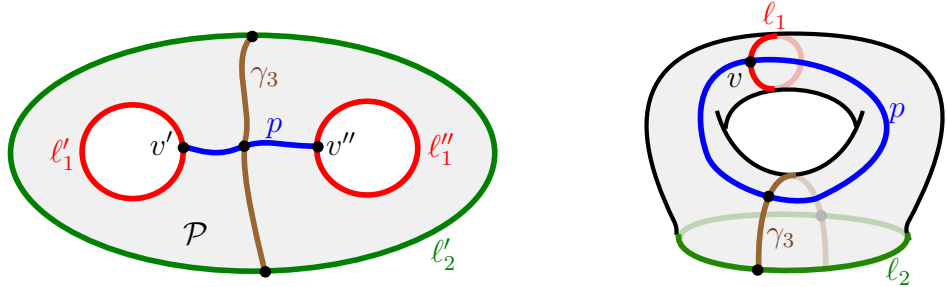}
            \caption{Left, the pair of pants cut by $\ell_1$ and $\ell_2$. Right, its image in $\Sigma$.}
            \label{fig:pair-of-pants}
          \end{figure}
          By the tightness of $\ell_2$, it is easily seen that $p$ is also a shortest path in $\Sigma\cut \ell_1$. We now choose $\ell_3$ to cross $p$ minimally among the third systoles disjoint from $\ell_1$ and crossing $\ell_2$ twice.
          \begin{claim}\label{clm:ell3-p}
            $\ell_3$ crosses $p$ exactly once.
          \end{claim}
          To see this, first note that the intersection of $\ell_3$ with $\cal P$ must be connected, since $\ell_3$ cuts $\ell_2$ twice and is disjoint from $\ell_1$. Denote by $\gamma_3$ this simple path. Since $\mathcal{P}\cut p$ is an annulus, if $\gamma_3$ did not cross $p$, then $\gamma_3$ would be homotopic to a subpath of $\ell_2$ and we could slide $\gamma_3$ on the other side of $\ell_2$ to remove the two crossings of $\ell_3$ and $\ell_2$. This would however contradict that $\ell_2$ and $\ell_3$ have geometric intersection number two. Now, if $\gamma_3$ cuts $p$ at least twice, then one of the components of $\gamma_3\setminus p$ must bound a bigon with a subpath of $p$. Since $p$ and $\ell_3$ are tight, the sides of the bigon have equal length and we may swap it to decrease the number of crossings of $\gamma_3$ and $p$. This would however contradict that $\ell_3$ cuts $p$ minimally. This concludes the proof of the claim.            

          It follows from Claim~\ref{clm:ell3-p} that $\alpha:=\ell_3\setminus p$ is a simple essential arc in $\Sigma\cut (\ell_1\cup p)$. Indeed, if $\alpha$ were not essential, then $\ell_3$ would be homotopic to $\ell_1$.
          \begin{claim}\label{clm:essential-arc}
            Any simple essential arc $\alpha$ connecting two copies of a vertex of $p$ in $\Sigma\cut (\ell_1\cup p)$ identifies in $\Sigma$ to a simple cycle $\ell$ that is not homotopic to $\ell_1, \ell_2$ or a point.
          \end{claim}
          If $\ell$ were homotopic to $\ell_1$, then these cycles would bound a cylinder traversed by a subpath of $p$. This would imply that $\alpha$ bounds a disk in $\Sigma\cut (\ell_1\cup p)$, contradicting that $\alpha$ is essential. Also observe that $p$ identifies to a simple loop in $\Sigma$, crossing $\ell$ once and disjoint from $\ell_2$. It ensues from the invariance of the algebraic intersection number that $\ell$ is neither contractible, nor homotopic to $\ell_2$.

We conclude from  Claims~\ref{clm:ell3-p} and~\ref{clm:essential-arc} that any shortest simple essential arc $\alpha$ connecting two copies of a vertex of $p$ in $\Sigma\cut (\ell_1\cup p)$ may serve as a third systole after identifying its endpoints in $\Sigma$.
        Computing the shortest path $p$ can be done in time $O(n\log n)$ using Dijkstra's algorithm. By Theorem~\ref{thm:shortest-arc}, a shortest essential arc with given source and target can be computed in time $O(n^2)$ and $p$ has $O(n)$ vertices, so this case takes time $O(n^3)$. 
        
        Now assume that $\ell_2$ is the brown cycle in Figure~\ref{fig:2crossings}. In this case, we consider the surface $\Sigma'=\Sigma\cut(\ell_1\cup\ell_2)$ with a pair of boundary components corresponding to $\ell_1$ and a pair of boundary components corresponding to $\ell_2$. Since $\ell_1$ and $\ell_2$ are both non-separating, $\Sigma'$ may have one or two connected components, but no more. Denote by $\delta$ and $\delta'$ the two copies of $\ell_2$ in $\Sigma'$ and by $\gamma_3$ and $\gamma_3'$ the two subpaths of $\ell_3$ cut by $\ell_2$. Remark that since $\ell_3$ is separating (it now corresponds to the green curve), the local orientation of $\ell_2$ and $\ell_3$ at their two crossings are reversed, implying that  $\gamma_3$ must have its endpoints on the same copy of $\ell_2$, say $\delta$, in $\Sigma'$. Hence, $\gamma_3'$ has its endpoints on the other copy $\delta'$. We also remark that $\gamma_3$ and $\gamma_3'$ must be essential arcs in $\Sigma'$; otherwise one could homotope an inessential arc of $\ell_3$ to remove its crossings with $\ell_2$, in contradiction with the fact that $\ell_2$ and $\ell_3$ are in minimal position. On the other hand, we argue below that any concatenation $\gamma'\cdot\gamma\inv$ in $\Sigma'$ of essential arcs $\gamma$ and $\gamma'$ with both  endpoints on $\delta$ and $\delta'$, respectively, cannot be homotopic to $\ell_1$, $\ell_2$ or a point. With the above remarks, this ensures that any shortest curve among such concatenations $\gamma'\cdot\gamma\inv$ may serve as a third systole. This leads to the following algorithm to compute a third systole. For every pair $u,v$ of distinct vertices on $\ell_2$, we compute a shortest essential arc $\gamma$ in $\Sigma'$ between the copies of $u,v$ on $\delta$. Similarly, we compute a shortest essential arc $\gamma'$ between the copies of $u,v$ on $\delta'$. We then return a shortest curve among such $\gamma'\cdot\gamma\inv$. We thus need to compute all the shortest essential arcs between any pair of vertices on $\delta$ and similarly for $\delta'$. By Theorem~\ref{thm:shortest-arcs}, this computation can be done in time $O(n^3)$. It remains to argue that the concatenations $\gamma'\cdot\gamma\inv$ cannot be homotopic to $\ell_1$, $\ell_2$ or a point. For this we rely on the following claim whose proof is deferred to the appendix.
\begin{claim}\label{cl:1}
  Any concatenation $\gamma'\cdot\gamma\inv$ in $\Sigma'$ of essential arcs $\gamma$ and $\gamma'$ with both endpoints on $\delta$ and $\delta'$, respectively, has geometric intersection number at least one with $\ell_2$.
\end{claim}
As a consequence, none of the $\gamma'\cdot\gamma\inv$ is contractible, or homotopic to $\ell_1$ or $\ell_2$. 
    \end{enumerate}
\end{proof}

\section{Discussion}\label{sec:discussion}

We have presented algorithms to compute the second and third systoles in a (weighted) combinatorial surface. A natural direction for future research would be to compute the $k$-th systole for any positive integer $k$. The computation of the $k$-th systole for $k>3$ using a similar approach as in the current paper would require to analyse the possible configurations of the first $k$ systoles. The case $k=3$ treated here shows that this analysis can be quite subtle. 

In particular, one would need to determine the number of self-crossings of the $k$-th systole. Some partial results for this number are known in the context of hyperbolic surfaces. Namely, Erlandsson and Parlier~\cite{ep-scgsi-20} proved that the shortest among the closed geodesics of a complete hyperbolic surface (with non-abelian fundamental group) with at least $k$ self-crossings has at most $O(k)$ self-crossings. If this bound was exactly $k$ and was true for other types of metrics, this would directly imply that the $k$-th systole has at most $k$ self-crossings. Some partial results in this direction were recently posted on arXiv~\cite{an-tsac-24}.

Moreover, one has to determine an upper bound for the number $n_{ik}$ of crossings between the $i$-th and $k$-th systoles, for $i<k$. Given $n_{ik}$ for all $i<k$, the $k$-th systole can be computed from the first $k-1$ systoles by choosing $n_{ik}$ crossing points on the $i$-th sytole for each $i<k$ and ordering all of them along the $k$-th systole (while keeping in mind the possible self-crossings). By connecting the crossing points in the given order with shortest (essential) arcs, one may form a candidate for the $k$-th systole. However, this approach would not be fixed parameter tractable with respect to $k$, since each systole may have a number of vertices of the same order as the whole combinatorial surface. 

\appendix

\section{Proof of Claim~\ref{cl:1}}

If $\Sigma'=\Sigma\cut(\ell_1\cup\ell_2)$ is not connected (as in Figure~\ref{fig:2crossings}, right) then $\gamma$ and $\gamma'$ belong to distinct components, hence are disjoint. It follows that $\gamma'\cdot\gamma\inv$ is a simple curve. If this curve had excess intersection with $\ell_2$, then there would be an embedded bigon bounded by $\gamma$ or $\gamma'$ on one side, and $\ell_2$ on the other side~\cite[Lem. 3.1]{hs-ics-85}. This would however contradict that $\gamma$ and $\gamma'$ are essential arcs in $\Sigma'$. 

If $\Sigma'$ is connected, then $\gamma$ and $\gamma'$ may a priori cross and $\gamma'\cdot\gamma\inv$ could be self-crossing and not homotopic to a simple curve. In this case, we cannot apply \cite[Lem. 3.1]{hs-ics-85} and we resort to another argument.
Recall that $S$ is the topological surface underlying $\Sigma$, as $\Sigma=(S,G)$.
We consider the universal cover $\tilde{S}$ of $S$. We endow $S$ with a hyperbolic metric with geodesic boundary, so that $\tilde{S}$ identifies with a subset of the Poincaré disk (See~\cite{katok1992},
especially Theorem 3.4.5 and Section 3.6). The lifts of $\ell_1$, or more precisely of $\ell_1^\infty: \R\to \R/\Z\overset{\ell_1}{\to}S$, form a disjoint union $\mathcal{L}_1$ of infinitely  many simple curves joining pairs of points ``at infinity'' on the boundary of the Poincaré disk.
Similarly, the lifts of $\ell_2$ form a disjoint union $\mathcal{L}_2$ of infinitely  many simple curves with endpoints at infinity. Since $\ell_1$ and $\ell_2$ are simple and disjoint, so are the curves in $\mathcal{L}_1\cup\mathcal{L}_2$.

Denote by $u$ and $v$ the crossing points of $\ell_2$ and $\gamma'\cdot\gamma\inv$ and consider some lift $\tilde{u}_1$ of $u$. It belongs to a lift, say $\lambda_1\in \mathcal{L}_2$ of $\ell_2$.
By lifting $\gamma'$ from $\tilde{u}_1$ we get a lifted curve $\tilde{\gamma}'$ whose other endpoint, say $\tilde{v}_1$, lies above $v$ on some lift $\lambda_2\in \mathcal{L}_2$.
\begin{claim}\label{cl:2}
  We have $\lambda_2\neq \lambda_1$.
\end{claim}
Suppose otherwise and denote by $\pi_1$ the subpath of $\lambda_1$ between $\tilde{u}_1$ and $\tilde{v}_1$. Then $\tilde{\gamma}' \cup \pi_1$ must be contained in a simply connected region whose interior avoids $\mathcal{L}_1\cup \mathcal{L}_2$, namely the intersection of all the ``half-planes'' containing $\tilde{\gamma}'$ and bounded by one of the curves in $\mathcal{L}_1\cup \mathcal{L}_2$. As a consequence, we can homotope $\tilde{\gamma}'$ to $\pi_1$ without crossing any curve in $\mathcal{L}_1$ or $\mathcal{L}_2$. It follows that $\gamma'$ is homotopic  in $S'$ (recall this is the surface $S$ cut by $\ell_1$ and $\ell_2$) to a path in $\delta'$. However, this would contradict the fact that $\gamma'$ is essential in $S'$, thus proving Claim~\ref{cl:2}.

Hence, $\tilde{v}_1$ belongs to a lift $\lambda_2\neq \lambda_1$ of $\ell_2$. We argue similarly for the lift of $\gamma$ from $\tilde{u}_1$ and for the lift of $\gamma\inv$ from $\tilde{v}_1$. We get by Claim~\ref{cl:2} that the first lift abuts to some lift $\lambda_0\neq \lambda_1$ of $\ell_2$, while the second lift abuts to some lift $\lambda_3\neq \lambda_2$ of $\ell_2$. Since $\gamma'\cdot\gamma\inv$ actually crosses $\ell_2$ at $u$ and $v$, $\lambda_0$ and $\lambda_2$ must lie on different sides of $\lambda_1$, and similarly $\lambda_1$ and $\lambda_3$ must lie on different sides of $\lambda_2$. The lifts $\lambda_0, \lambda_1, \lambda_2, \lambda_3$ are thus parallel "lines" in the plane appearing in this order (their endpoints constitute a nested set of parentheses). By induction on the number of lifts of $\gamma'\cdot\gamma\inv$, we obtain an infinite sequence $\{\lambda_i\}_{i\in\Z}$ of lifts of $\ell_2$ with a lift $\mu$ of $(\gamma'\cdot\gamma\inv)^\infty$ that crosses each $\lambda_i$ exactly once. In particular, the endpoints at infinity of $\mu$ are interleaved with the endpoints of $\lambda_0$. Since these endpoints are independent of the representative of the curves in their free homotopy class, it ensues that any two curves homotopic to respectively $\ell_2$ and $\gamma'\cdot\gamma\inv$ cross at least once. This finishes the proof of Claim~\ref{cl:1}.


\end{document}